\documentclass[pdflatex,sn-mathphys-num]{sn-jnl}

\usepackage{graphicx}%
\usepackage{multirow}%
\usepackage{amsmath,amssymb,amsfonts}%
\usepackage{amsthm}%
\usepackage{mathrsfs}%
\usepackage[title]{appendix}%
\usepackage{xcolor}%
\usepackage{textcomp}%
\usepackage{manyfoot}%
\usepackage{booktabs}%
\usepackage{algorithm}%
\usepackage{algorithmicx}%
\usepackage{algpseudocode}%
\usepackage{listings}%
\usepackage{amssymb}
\usepackage{todonotes}

\usepackage{subfigure}

\DeclareMathOperator{\argmax}{argmax}

\newcommand{\be}{\begin{equation}}
\newcommand{\ee}{\end{equation}}

\theoremstyle{thmstyleone}%
\newtheorem{theorem}{Theorem}
\newtheorem{proposition}{Proposition}%

\theoremstyle{thmstyletwo}%
\newtheorem{remark}{Remark}%

\newtheorem{lemma}{Lemma}

\theoremstyle{thmstylethree}%

\begin{document}

\title[Peak Prevalence in  SIR Epidemics]{Approximating Peak Prevalence in Multistage SIR Epidemics}


\author[1,2]{\fnm{Denis} \sur{Tverskoi}}\email{tverskoi.1@osu.edu}
\equalcont{These authors contributed equally to this work.}

\author[1]{\fnm{Andrew} \sur{Gothard}}\email{gothard.17@buckeyemail.osu.edu}

\author*[1,3]{\fnm{Grzegorz A.} \sur{Rempala}}\email{rempala.3@osu.edu}
\equalcont{These authors contributed equally to this work.}

\affil*[1]{\orgdiv{Division of Biostatistics, College of Public Health}, \orgname{The Ohio State University}, \orgaddress{\city{Columbus}, \postcode{43210}, \state{OH}, \country{USA}}}

\affil[2]{\orgdiv{Department of Anthropology, College of Arts and Sciences}, \orgname{The Ohio State University}, \orgaddress{\city{Columbus}, \postcode{43210}, \state{OH}, \country{USA}}}

\affil[3]{\orgdiv{Department of Mathematics, College of Arts and Sciences}, \orgname{The Ohio State University}, \orgaddress{\city{Columbus}, \postcode{43210}, \state{OH}, \country{USA}}}


\abstract{
Estimating peak prevalence is a central problem in epidemic modeling because it determines the period of greatest infectious burden and is closely linked to health-care demand. In multistage SIR models, however, peak prevalence is generally less tractable than in the classical model with exponentially distributed infectious periods. Motivated by the use of weighted infectious-stage aggregates as surrogates for prevalence, we investigate the relationship between the prevalence peak and the maximum of a weighted stage functional in deterministic SI$(k)$R epidemic models.
We show that this relationship depends critically on how the stage-progression rate is scaled as the number of infectious stages increases. Under naive scaling, in which the progression rate remains fixed, the weighted peak is asymptotically equivalent to the prevalence peak and the commonly used factor-two approximation fails. Under Erlang scaling, which preserves the mean infectious period, the multistage model converges to a delay formulation in which prevalence and the weighted stage functional become unweighted and triangularly weighted moving averages of incidence.

This limiting representation provides a theoretical basis for the factor-two approximation and identifies the regimes in which it is accurate. It also explains why this approximation deteriorates as epidemic waves become more sharply peaked. We derive analytical error bounds and develop curvature-based and parameter-based corrections that substantially improve accuracy. Numerical studies confirm these improvements across a broad range of epidemiological parameters. Overall, the results show when weighted-stage peaks can be used reliably as proxies for peak prevalence and how the resulting estimates can be refined when the standard approximation loses accuracy.
}




\maketitle

\tableofcontents

\section{Introduction}
Estimating the maximum prevalence of infection during an epidemic is a central problem in mathematical epidemiology \citep{anderson1991infectious,KeelingRohani2008}. Much of our understanding of epidemic prevalence and its determinants comes from compartmental models, which are widely used to quantify disease burden, evaluate intervention strategies, and characterize the temporal course of outbreaks. Among the many epidemiological quantities that can be derived from such models, the prevalence peak is of particular practical importance because it determines the period of greatest simultaneous infectious burden and thus directly influences health-care demand and public-health response \citep{anderson1991infectious,ferg2020}. The magnitude of this peak is directly related to clinical burden, workforce disruption, and pressure on health-care capacity, making it a natural measure of epidemic severity.
In practice, however, prevalence is often difficult to observe directly. Case reports and surveillance systems more naturally record incidence, while point prevalence is affected by under-ascertainment, reporting delays, variation in testing effort, and asymptomatic or mildly symptomatic infection. This motivates analytic methods that recover, approximate, or bound prevalence peaks from quantities that are more naturally connected to incidence and susceptible depletion.
The classical susceptible--infectious--removed (SIR) framework of Kermack and McKendrick remains the starting point for the study of epidemic prevalence \citep{kermack1927contribution,kermack1991contributions,hethcote2000mathematics,anderson1991infectious}. In this formulation, $S(t)$, $I(t)$, and $R(t)$ denote the fractions of the population that are susceptible, infectious, and removed (recovered or otherwise no longer infectious), respectively. The corresponding peak prevalence is
\[
I_{\max}=\max_{t\ge 0} I(t).
\]
For the standard model with an exponentially distributed infectious period, the prevalence peak occurs when the susceptible fraction crosses the threshold \(S=1/R_0\), and the corresponding peak prevalence can be expressed in terms of the classical epidemic invariant and final-size relation. This simple characterization relies heavily on the memoryless nature of the infectious period and is generally lost when more realistic infectious-period distributions are considered.
A common approach to incorporating non-exponential infectious periods is to partition the infectious class into multiple sequential stages. Such multi-stage or linear-chain models replace the memoryless infectious period by an Erlang distribution and provide a finite-dimensional approximation to age-of-infection and renewal formulations \citep{macdonald1978time,smith2011distributed,hurtado2019generalizations,champredon2018equivalence,diekmann2013mathematical,tverskoi2025flu}. 

Such multistage models are of interest not only because they provide more realistic infectious-period distributions, but also because they form a bridge between finite-dimensional compartmental models and age-of-infection or renewal formulations. As the number of infectious stages increases, the Erlang distribution becomes increasingly concentrated around its mean, and the resulting epidemic dynamics approach those of a delay-type epidemic model. This observation naturally motivates the study of the large-$k$ regime.

For any fixed finite number of stages $k$, the quantities like $I_{\max}$  can be computed directly from numerical solutions of the governing ODE system, and for the relatively small values of $k$ commonly used in applications, such calculations are routine. Consequently, our interest in the large-$k$ limit is not primarily computational. Rather, the asymptotic analysis reveals structural relationships that are difficult to discern from finite-dimensional calculations alone. In particular, it explains the emergence of simple approximations for peak prevalence, identifies the limiting relationship between prevalence and weighted incidence, helps to identify when simple approximations  are expected to be accurate, and leads to explicit correction formulas that remain useful even for moderate values of $k$.

The asymptotic perspective developed below is particularly useful for
understanding the relationship between prevalence and certain weighted
occupancy measures of the infectious population. To establish the framework
for the asymptotic analysis, we first introduce a general $k$-stage epidemic
model and define an associated weighted aggregate of infectious-stage
occupancies whose peak admits an explicit characterization. This
representation forms the basis for all subsequent asymptotic analysis. Specifically,  we study the deterministic $k$-staged epidemic model $SI(k)R$:

\begin{align}
\dot{S}^{(k)} &= -\beta S^{(k)} I^{(k)}, \nonumber\\
\dot{I}_1^{(k)} &= \beta S^{(k)} I^{(k)}-\delta I^{(k)}_1, \label{eq:SIkR}\\
\dot{I}_i^{(k)} &= \delta I^{(k)}_{i-1}-\delta I^{(k)}_i,
\qquad i=2,\dots,k, \nonumber\\
\dot{R}^{(k)} &= \delta I^{(k)}_k, \nonumber
\end{align}
where
\[
I^{(k)}(t)=\sum_{i=1}^k I^{(k)}_i(t).
\]
The initial conditions are
\begin{equation}
\label{eq:ic_SIkR}
S^{(k)}(0)=1-\varepsilon>0, \qquad
I^{(k)}_1(0)=\varepsilon>0, \qquad
I^{(k)}_i(0)=0,\ i=2,\dots,k, \qquad
R^{(k)}(0)=0.
\end{equation}
Although the prevalence \(I^{(k)}(t)\) is the primary epidemiological quantity of interest, a key role in our analysis is played by the auxiliary weighted stage aggregate
\begin{equation}
\label{eq:Vdef}
V^{(k)}(t)=\sum_{i=1}^{k}(k-i+1)I_i^{(k)}(t),
\end{equation}
which may be interpreted as the total remaining infectious-stage mass. Closely related weighted occupancy functionals have appeared previously in the analysis of multistage epidemic models. In particular, \citet{islier2020finalsize} employed \(V^{(k)}\) in a stochastic SI\((k)\)R model and observed that its maximum provides a useful proxy for the prevalence peak.

A key feature of \(V^{(k)}\) is the telescoping identity
\[
\frac{\dot V^{(k)}(t)}{k}
=
\bigg(\beta S^{(k)}(t)-\frac{\delta}{k}\bigg)I^{(k)}(t).
\]
Thus, whenever the maximizer \(t_V\) of \(V^{(k)}\) is interior and \(I^{(k)}(t_V)>0\), it is determined entirely by the susceptible trajectory:
\(
S^{(k)}(t_V)=\frac{\delta}{k\beta}.
\)
This relationship connects the weighted peak to susceptible depletion and reconstructed incidence, without requiring direct observation of the full infectious-stage distribution. Since $V^{(k)}$ grows linearly with the number of stages, it is convenient to work with the normalized quantity $W^{(k)}(t) = V^{(k)}(t)/k$. The corresponding peak value,
\be\label{eq:Wmax} 
W_{\max}^{(k)}=\max_{t\ge 0}\frac{V^{(k)}(t)}{k},
\ee
 is the main object of study in the paper and forms the basis for constructing approximations to the prevalence peak.


The main objective of this paper is to determine when the weighted peak $W_{\max}^{(k)}$ provides a reliable proxy for the prevalence peak $I_{\max}^{(k)}$ and how this relationship depends on the scaling of the stage progression rate. Previous work \citep{islier2020finalsize} heuristically proposed the approximation
\be
\label{eq:sc1}
I_{\max}^{(k)}\sim  2W_{\max}^{(k)}.
\ee
However, the mathematical basis for this approximation, its accuracy, and
the range of parameter regimes in which it is valid have remained unclear.

The analysis presented in this paper shows that when $\delta$ is held fixed
as $k\to\infty$ (the \emph{naive scaling}), the epidemic must drive
$S^{(k)}(t_V)$ to order $k^{-1}$, while the occupied infectious stages
remain concentrated near the beginning of the chain. As a result, the
weights defining $V^{(k)}$ are nearly constant over the relevant stages,
leading to the asymptotic relation
\be
\label{eq:sc2}
I_{\max}^{(k)} \sim W_{\max}^{(k)}.
\ee
Thus, under naive scaling, the factor-two approximation is not asymptotically valid.

A qualitatively different picture emerges under the \emph{Erlang scaling} $\delta=k\gamma$, which preserves the mean infectious period $\tau = 1/\gamma$. In this case the infectious population spreads across the entire stage structure, and the normalized process $W^{(k)}$ converges to a limiting weighted-incidence functional $W$. This limit provides a natural connection between the finite-dimensional multistage model and a delay (or renewal) formulation. In particular, prevalence becomes a moving average of incidence over the infectious period, whereas $W$ corresponds to a triangularly weighted average over the same interval. For large $k$, this representation naturally yields the approximation \eqref{eq:sc1}, whose accuracy we analyze and then systematically refine through higher-order corrections derived from the local geometry of the incidence curve near its peak.

The paper is structured as follows. In Section~\ref{sec:2}, we derive an explicit expression for $V^{(k)}_{\max}$, valid for arbitrary $k$, which serves as the starting point for the subsequent asymptotic analysis. We then examine two distinct scaling regimes. In Section~\ref{sec:2a}, we show that under the \emph{naive scaling}, the weighted peak is asymptotically equivalent to the prevalence peak. In Section~\ref{sec:3}, we consider the \emph{Erlang scaling} and establish a infinite-stage limit that links the multistage model to a delay epidemic formulation (Section~\ref{ssec:lim}). This limit justifies  the approximation $I_{\max}^{(k)}\approx 2W_{\max}^{(k)}$, introduced in Section~\ref{ssec:la}. We analyze the accuracy of this approximation systematically in Section~\ref{ssec:accuracy} and use the same limiting representation as the foundation for the refined approximations developed in Section~\ref{ssec:refined}. From an applications perspective, this section contains perhaps the paper's most practically relevant contributions, as it provides a hierarchy of increasingly accurate approximations based on local properties of the incidence curve. Section~\ref{sec:summary} concludes the paper with a summary of the main results and a discussion of their implications. Technical details and proofs omitted from the main text are collected in the Appendix.

\section{Preliminaries}\label{sec:2}
We begin by focusing on the weighted stage aggregate $V^{(k)}$. A key advantage of this quantity is that its dynamics admit a simple telescoping representation, allowing its extrema to be characterized directly in terms of the susceptible trajectory. The next result gives an explicit formula for $V^{(k)}_{\max}$ and identifies the susceptible level at which the maximum occurs.

\begin{proposition}\label{prop:1} Consider the ODE system \eqref{eq:SIkR} with  $k\ge 1$ and with the initial condition \eqref{eq:ic_SIkR}.
Assume that $1 - \varepsilon > \frac{\delta}{k \beta}$. Then, 
\be
\label{eq:Vmax}
V^{(k)}_{\max} = k-\frac{\delta}{\beta} - \frac{\delta}{\beta} \ln \bigg( \frac{\beta k(1-\varepsilon)}{\delta} \bigg).
\ee
\end{proposition}

\begin{proof}
Differentiating~\eqref{eq:Vdef},
\[
\dot{V}^{(k)}(t)
=
\sum_{i=1}^k (k-i+1)\dot{I}^{(k)}_i(t).
\]
Using equations for $I_i$ in \eqref{eq:SIkR},
\be
\label{eq:Vprime}
\dot{V}^{(k)} = k \big( \beta S^{(k)}I^{(k)}-\delta I^{(k)}_1 \big)
+\sum_{i=2}^k (k-i+1)(\delta I^{(k)}_{i-1}-\delta I^{(k)}_i) = (k\beta S^{(k)}-\delta)I^{(k)}.
\ee
Hence the extremum of $V^{(k)}$ occurs at time $t_V^{(k)}$ satisfying
\be
\label{eq:Sv}
S^{(k)}(t_V^{(k)}) = \frac{\delta}{k\beta}.
\ee
From~\eqref{eq:SIkR} and~\eqref{eq:Vprime}, we get
\be
\frac{dV^{(k)}}{dS^{(k)}} = \frac{\dot{V}^{(k)}}{\dot{S}^{(k)}} =
-k+\frac{\delta}{\beta S^{(k)}}.
\label{eq:dVdS}
\ee
Integrating and using initial conditions $V^{(k)}(0)=k\varepsilon$ and $S^{(k)}(0)=1-\varepsilon$, we obtain:
\be
\label{eq:VS}
V^{(k)} \big(S^{(k)} \big) = k \big(1-S^{(k)} \big) -\frac{\delta}{\beta} \ln \bigg(\frac{1-\varepsilon}{S^{(k)}} \bigg).
\ee
Substituting \eqref{eq:Sv} into \eqref{eq:VS} gives \eqref{eq:Vmax}.
\end{proof}

Note that 
the assumed condition
\[
1-\varepsilon>\frac{\delta}{k\beta}
\]
ensures that the initial susceptible fraction exceeds the critical value
\(S^{(k)}=\delta/(k\beta)\) at which \(V^{(k)}\) attains an extremum. Since
\(S^{(k)}(t)\) decreases monotonically over time, this guarantees that the
trajectory crosses the threshold \(S^{(k)}=\delta/(k\beta)\) during the
epidemic, so that \(V^{(k)}\) possesses an interior maximum. Indeed, if
\(1-\varepsilon\le \delta/(k\beta)\), then \(\dot V^{(k)}(t)\le 0\) for all
\(t\ge0\), and the maximum is attained at the initial time,
\(V^{(k)}_{\max}=V^{(k)}(0)=k\varepsilon\).

\section{Naive scaling}\label{sec:2a}
In this section, we consider the \emph{naive scaling} assuming that $\delta$ is held fixed as $k\to\infty$. Under this regime, we show  that the weighted peak is asymptotically equivalent to the prevalence peak.

We first establish a lemma that provides lower and upper bounds for the time $t_V^{(k)}$ at which $V^{(k)}$ attains its maximum.

\begin{lemma}\label{lem:1} Under the assumptions of Proposition~\ref{prop:1},
 \be
\label{eq:tv}
\frac{1}{\beta} \ln\bigg[ \frac{(1-\varepsilon)\beta}{\delta} k \bigg] \leq t_V^{(k)} \leq \frac{1}{\varepsilon \beta} \ln\bigg[ \frac{(1-\varepsilon)\beta}{\delta} k \bigg]. 
\ee
\end{lemma}

\begin{proof}
Since $\dot{S}^{(k)} = -\beta S^{(k)}I^{(k)}$, we have
\[
S^{(k)}(t) = (1-\varepsilon)e^{-\beta \int_{0}^{t} I^{(k)}(u) du}.
\]
Calculated at $t=t_V^{(k)}$ and noting that for all $t$, $I^{(k)}(t) \leq 1$, it gives
\[
\frac{1}{\beta} \ln\bigg[ \frac{(1-\varepsilon)\beta}{\delta} k \bigg] = \int_{0}^{t_V^{(k)}} I^{(k)}(u) du \leq t_V^{(k)}.
\]
This establishes the lower bound. For the upper bound, note that $V^{(k)}(t)$ is a non-decreasing function of $t$ if $t \leq t_V^{(k)}$. Therefore,
\[
\forall t \in [0,t_V^{(k)}]: k \varepsilon =V^{(k)}(0) \leq V^{(k)}(t).
\]
In addition, 
\[
V^{(k)}(t) = \sum_{i=1}^{k} (k-i+1) I^{(k)}_i(t) \leq k I^{(k)}(t).
\]
Together, it gives
\[
\forall t \in [0,t_V^{(k)}]: I^{(k)}(t) \geq \varepsilon.
\]
This also means that 
\[
\forall t \in [0,t_V^{(k)}]: \dot{S}^{(k)} = - \beta S^{(k)} I^{(k)} \leq - \beta \varepsilon S^{(k)},
\]
which gives
\[
\forall t \in [0,t_V^{(k)}]: S^{(k)}(t) \leq (1-\varepsilon) e^{-\beta \varepsilon t}.
\]
Evaluating this inequality at $t=t_V^{(k)}$, we obtain:
\[
\frac{\delta}{k \beta} \leq (1-\varepsilon) e^{-\beta \varepsilon t_V^{(k)}}\quad \text{yielding}\quad \text t_V^{(k)} \leq \frac{1}{\varepsilon \beta} \ln\bigg[ \frac{(1-\varepsilon)\beta}{\delta} k \bigg].
\]
\end{proof}
Next, using the upper bound for $t_V^{(k)}$ from Lemma~\ref{lem:1}, we establish a lower bound for $I^{(k)}_{\max}$.

\begin{lemma}\label{lem:2} 
Under the assumptions of Proposition~\ref{prop:1}, 
\be
\label{eq:Im}
I^{(k)}_{\max} \geq 1 - \frac{\delta}{\beta} \cdot \frac{1}{k} - \frac{\delta^k}{k!} \bigg[ \frac{1}{\varepsilon \beta} \ln \bigg( \frac{(1-\varepsilon)\beta}{\delta} k \bigg) \bigg]^k.
\ee
\end{lemma}

\begin{proof}
We can derive the following estimates:
\[
I^{(k)}_1(t) \leq 1,
\]
\[
\dot{I}^{(k)}_2(t) = \delta I^{(k)}_1 - \delta I^{(k)}_2 \leq \delta I_1^{(k)}(t) \leq \delta \quad\text{implying}\quad I^{(k)}_2(t) \leq \delta t,
\]
\[
I^{(k)}_3(t) \leq \frac{(\delta t)^2}{2},..., I^{(k)}_k(t) \leq \frac{(\delta t)^{k-1}}{(k-1)!}.
\]
Therefore,
\[
R^{(k)}(t) = \delta \int_{0}^{t} I^{(k)}_k(\tau) d\tau \leq \frac{(\delta t)^k}{k!}.
\]
As a result,
\[
I^{(k)}_{\max} \geq I^{(k)}(t_V^{(k)}) = 1 - S^{(k)}(t_V^{(k)}) - R^{(k)}(t_V^{(k)}) \geq 1 - \frac{\delta}{\beta} \cdot \frac{1}{k} - \frac{(\delta t_V^{(k)})^k}{k!}.
\]
Applying Lemma~\ref{lem:1}, we obtain the lower bound for $I_{\max}$.
\end{proof}

Finally, using the lower bound for $I^{(k)}_{\max}$ from Lemma~\ref{lem:2} together with the explicit formula for $V^{(k)}_{\max}$ in~\eqref{eq:Vmax}, we characterize the asymptotic behavior of $\frac{W^{(k)}_{\max}}{I^{(k)}_{\max}} = \frac{V^{(k)}_{\max}}{k I^{(k)}_{\max}}$.

\begin{theorem}
Under the assumptions of Proposition~\ref{prop:1}, 
\[
\frac{W^{(k)}_{\max}}{I^{(k)}_{\max}} = 1 - \frac{\delta}{\beta} \cdot \frac{\ln(k)}{k} + O(1/k), \text{ } k \to \infty.
\]
\end{theorem}

\begin{proof}
First note that $I^{(k)}_{\max} \leq 1$. Therefore,
\[
\frac{W^{(k)}_{\max}}{I^{(k)}_{\max}} \geq W^{(k)}_{\max} \geq 1 - \frac{\delta}{\beta} \cdot \frac{1}{k} - \frac{\delta}{\beta} \cdot \frac{1}{k} \ln\bigg[ \frac{(1-\varepsilon)\beta}{\delta} k \bigg].
\]
On the other hand, assume that $k$ is sufficiently large so that $1-\frac{\delta}{\beta} \cdot \frac{1}{k} - Z(k) > 0$, where
\[
Z(k) = \frac{1}{k!} \bigg[ \frac{\delta}{\varepsilon \beta} \ln \bigg( \frac{(1-\varepsilon)\beta}{\delta} k \bigg) \bigg]^k = o(1/k).
\]
Then, employing Lemma~\ref{lem:2} we obtain:
\[
\frac{W^{(k)}_{\max}}{I^{(k)}_{\max}} \leq \frac{W^{(k)}_{\max}}{1-\frac{\delta}{\beta} \cdot \frac{1}{k} - Z(k)} = W^{(k)}_{\max} \bigg( 1 + \frac{\delta}{\beta} \cdot \frac{1}{k} + o(1/k) \bigg) = 1 - \frac{\delta}{\beta} \cdot \frac{1}{k} \ln\bigg[ \frac{(1-\varepsilon)\beta}{\delta} k \bigg]  + o(1/k).
\]
As a result, for a sufficiently large $k$,
\[
- \frac{\delta}{\beta} \cdot \frac{1}{k} \bigg(1+ \ln\bigg[ \frac{(1-\varepsilon)\beta}{\delta} \bigg] \bigg) \leq \frac{W^{(k)}_{\max}}{I^{(k)}_{\max}} - 1 +  \frac{\delta}{\beta} \cdot \frac{\ln(k)}{k} \leq o(1/k) - \frac{\delta}{\beta} \cdot \frac{1}{k} \ln\bigg[ \frac{(1-\varepsilon)\beta}{\delta} \bigg].
\]
\end{proof}

Thus, under the naive scaling the weighted peak $W_{\max}^{(k)}$ does provide a reliable proxy for the prevalence peak, but with the asymptotic proportionality constant equal to 1, not 2. This shows that the heuristic approximation $I_{\max}^{(k)}\approx 2W_{\max}^{(k)}$ cannot be justified by simply taking $k\to\infty$ with $\delta$ fixed. The factor of two must instead be tied to the Erlang scaling considered in the next section.

\section{Erlang scaling}\label{sec:3} 
In this section, we consider the \emph{Erlang scaling} $\delta=k\gamma$ with $k\to\infty$ and $R_0=\beta/\gamma$. This scaling preserves the mean infectious period $\tau = 1/\gamma$ as the number of stages increases. We first establish a infinite-stage limit that links the multistage model to a delay epidemic formulation. Using this limit, we justify the approximation $I_{\max}^{(k)}\approx 2W_{\max}^{(k)}$. We then analyze the accuracy of this approximation and use the same limiting representation as the basis for the refined approximations.

\subsection{Infinite-stage limit}
\label{ssec:lim}
Here, we show that the trajectory $W^{(k)}$ converges as $k \to \infty$, to a limiting weighted-incidence functional $W$. This limit provides a natural connection between the multistage model~\eqref{eq:SIkR} and a delay formulation. Specifically, we show that prevalence becomes a moving average of incidence over the infectious period, whereas $W$ corresponds to a triangularly weighted average over the same interval.

\begin{theorem}[Infinite-stage limit of the $SI(k)R$ infectious profile]
\label{thm:1} Consider the ODE system \eqref{eq:SIkR} with the initial conditions \eqref{eq:ic_SIkR}.
Let $A^{(k)}(t)=-\dot{S}^{(k)}$ denote the incidence of entering  the first infectious stage. Suppose that, as $k\to \infty$,
$A^{(k)}\to A$ locally uniformly on $[0,\infty)$, where $A$ is continuous. Fix  $x\in(0,1)$ and set $i_k=\lfloor kx\rfloor$. If
$t>\tau x$, then
\[
k I_{i_k}^{(k)}(t)
\longrightarrow
\tau A (t- \tau x).
\]
If $t< \tau x$, then
\[
k I_{i_k}^{(k)}(t)\longrightarrow 0.
\]
The initial mass $\varepsilon$ is transported along the characteristic
$x=t / \tau$  in the weak convergence  sense:
\[
k I_{\lfloor kx\rfloor}^{(k)}(t)\,dx
\Longrightarrow
\varepsilon \delta_{t/ \tau}(dx)
+
\tau A (t- \tau x)
\mathbf 1_{\{0<x<\min(1,t / \tau)\}}\,dx .
\]
Moreover, for every continuity point of the limiting expressions,
\[
I^{(k)}(t)\longrightarrow
I(t)
=
\varepsilon \mathbf 1_{\{t<\tau\}}
+
\int_{(t-\tau)_+}^{t} A(u)\,du,
\]
and
\[
W^{(k)}(t)
\longrightarrow
W(t)
=
\varepsilon \bigg(1- \frac{t}{\tau} \bigg)_+
+
\int_{(t-\tau)_+}^{t}
\bigg(1- \frac{t-u}{\tau} \bigg) A(u)\,du .
\]
\end{theorem}

\begin{proof} By direct verification, or alternatively by the method of variation of parameters applied to \eqref{eq:SIkR} and \eqref{eq:ic_SIkR}, we obtain

\[
I_i^{(k)}(t)
=
\varepsilon e^{- \frac{k t}{\tau}} \frac{(k t)^{i-1}}{\tau^{i-1} (i-1)!}
+
\int_0^t
A^{(k)}(u)
e^{-\frac{k (t-u)}{\tau}}
\frac{(k (t-u))^{i-1}}{\tau^{i-1} (i-1)!}\,du .
\] 
Let
\[
p_{k,i}(s)
=
e^{- \frac{k s}{\tau}} \frac{(k s)^{i-1}}{\tau^{i-1} (i-1)!}
\]
This is the Poisson probability
\[
p_{k,i}(s)=\mathbb P\{N_{k s / \tau}=i-1\}.
\]
Thus
\[
I_i^{(k)}(t)
=
\varepsilon p_{k,i}(t)
+
\int_0^t A^{(k)}(u)p_{k,i}(t-u)\,du .
\]
Consider  $i=i_k=\lfloor kx\rfloor$.
For the integral term, write $r=t-u$. Then
\[
k\int_0^t A^{(k)}(u)p_{k,i_k}(t-u)\,du
=
k\int_0^t A^{(k)}(t-r)p_{k,i_k}(r)\,dr .
\]
The kernel $K_k(r):=k p_{k,i_k}(r)$ has total mass
\[
\int_0^\infty K_k(r)\,dr
=
k\int_0^\infty e^{-\frac{kr}{\tau}}
\frac{(k r)^{i_k-1}}{\tau^{i_k-1}(i_k-1)!}\,dr
=
\tau.
\]
Moreover, by Stirling's formula,
\[
(i_k-1)!
\sim
\sqrt{2\pi i_k}\left(\frac{i_k}{e}\right)^{i_k},
\]
so for $r>0$,
\[
p_{k,i_k}(r)
\sim
\frac{1}{\sqrt{2\pi kx}}
\exp\left\{
-k\left[
x\ln\left(\frac{\tau x}{r}\right)-x+ \frac{r}{\tau}
\right]
\right\}.
\]
The expression in the bracket is nonnegative and vanishes only at
\(
r=\tau x.
\)
Therefore $K_k(r)$ concentrates at $r=\tau x$ and has total mass
$\tau$. Hence, for continuous bounded $f$,
\[
\int_0^\infty f(r)K_k(r)\,dr
\to
\tau f(\tau x).
\]
Applying this with $f(r)=A(t-r)\mathbf 1_{\{0<r<t\}}$ gives
\[
k\int_0^t A^{(k)}(t-r)p_{k,i_k}(r)\,dr
\to
\tau 
A(t-\tau x)
\mathbf 1_{\{\tau x<t\}}.
\]
The initial term is $k\varepsilon p_{k,i_k}(t)$. By the same Stirling expansion, this converges to zero whenever
$x\neq \frac{t}{\tau}$, since then the bracket expression above  is strictly positive.
At $x=\frac{t}{\tau}$ the pointwise expression is of order $\sqrt{k}$,
reflecting concentration of the initial cohort. In weak form this term
converges to $\varepsilon\delta_{t / \tau}(dx)$ (see Appendix~\ref{app:wc}). This proves the stated infinite-stage profile.

Next consider the total prevalence. Since
\[
\sum_{i=1}^k p_{k,i}(s)
=
\mathbb P\{N_{k s / \tau}\le k-1\},
\]
we have
\[
I^{(k)}(t)
=
\varepsilon \mathbb P\{N_{k t / \tau}\le k-1\}
+
\int_0^t
A^{(k)}(u)
\mathbb P\{N_{k(t-u) / \tau}\le k-1\}\,du .
\]
By the law of large numbers for Poisson variables (see, for instance, \cite{Ross2014}),
\[
\frac{N_{k s / \tau}}{k}\to \frac{s}{\tau}
\]
in probability. Hence
\[
\mathbb P\{N_{k s / \tau}\le k-1\}
\to
\mathbf 1_{\{s / \tau <1\}}
\]
at continuity points. Since 
$A^{(k)}$ is locally bounded, so dominated convergence (i.e. uniform integrability of $\frac{N_{k s / \tau}}{k}$)  gives
\[
I^{(k)}(t)\to
\varepsilon\mathbf 1_{\{t<\tau\}}
+
\int_0^t
A(u)\mathbf 1_{\big\{\frac{t-u}{\tau}<1\big\}}\,du
=
\varepsilon\mathbf 1_{\{t<\tau\}}
+
\int_{(t-\tau)_+}^{t}A(u)\,du
=
I(t).
\]

Finally, using the cohort representation of \(I_i^{(k)}(t)\), we may write
\begin{align*}
W^{(k)}(t)
&=
\sum_{i=1}^k
\left(1-\frac{i-1}{k}\right)I_i^{(k)}(t)
\\
&=\varepsilon \sum_{i=1}^k
\left(1-\frac{i-1}{k}\right)p_{k,i}(t)
+
\int_0^t A^{(k)}(u)
\sum_{i=1}^k
\left(1-\frac{i-1}{k}\right)p_{k,i}(t-u)\,du.
\end{align*}
For any \(s\geq 0\),
\[
\sum_{i=1}^k
\left(1-\frac{i-1}{k}\right)p_{k,i}(s)
=
\mathbb E\left[
\left(1-\frac{N_{k s / \tau}}{k}\right)_+
\right].
\]
Indeed, the event \(N_{k s / \tau}=i-1\) corresponds to an individual
being in stage \(i\), and the positive part accounts for the fact that
only stages \(1,\ldots,k\) contribute to \(V^{(k)}\).
Since \(\frac{N_{k s / \tau}}{k} \to \frac{s}{\tau} \) in probability and
\(x\mapsto(1-x)_+\) is bounded and continuous, bounded convergence
implies
\[
\sum_{i=1}^k
\left(1-\frac{i-1}{k}\right)p_{k,i}(s)
\to
\bigg( 1- \frac{s}{\tau} \bigg)_+ .
\]
Applying the preceding limit with \(s=t\), the initial cohort contributes
$\varepsilon \big(1 - \frac{t}{\tau} \big)_+$ to the limit. For individuals infected at time \(u\), their infection age at time \(t\)
is \(s=t-u\). Hence the same argument gives
\[
\sum_{i=1}^k
\left(1-\frac{i-1}{k}\right)p_{k,i}(t-u)
\to
\bigg( 1- \frac{t-u}{\tau} \bigg)_+ .
\]
Assuming \(A^{(k)}\to A\) pointwise, with the sequence \(A^{(k)}\)
locally dominated by an integrable bound, dominated convergence yields
\[
\int_0^t A^{(k)}(u)
\sum_{i=1}^k
\left(1-\frac{i-1}{k}\right)p_{k,i}(t-u)\,du
\to
\int_0^t A(u) \bigg( 1- \frac{t-u}{\tau} \bigg)_+ \,du .
\]
Therefore,
\[
W^{(k)}(t)
\to
W(t)
=
\varepsilon \bigg( 1 - \frac{t}{\tau} \bigg)_+
+
\int_0^t A(u) \bigg( 1- \frac{t-u}{\tau} \bigg)_+ \,du .
\]
\end{proof}

Theorem~\ref{thm:1} provides a direct link between the stage-structured Erlang model and its limiting equation, showing that both $I$ and $W$ are determined by the same underlying incidence process, albeit through different weighting mechanisms. Specifically in the Erlang scaling limit, prevalence and the weighted stage aggregate correspond to two different averages of the same incidence trajectory over the infectious period: a uniform average for prevalence $I$ and a triangular average for $W$.

\begin{remark}\label{rem:1} Observe that since we expect that under mild regularity conditions  (see, for instance, \cite{diekmann2013mathematical}) the assumption \(A^{(k)}\to A\) holds with   \(A(u)=-\dot S(u)\), for some function $S$, we have  
\[
I(t)
=
\varepsilon\mathbf 1_{\{t<\tau\}}
+
\int_{(t-\tau)_+}^{t}A(u)\,du
=
\varepsilon\mathbf 1_{\{t<\tau\}}
+
S\!\left((t-\tau)_+\right)-S(t).
\]
Therefore, as $k \to \infty$ and \(A^{(k)}\to A\), the ODE system~\eqref{eq:SIkR} transforms into the following delay differential equation:

\be
\label{eq:DDE1}
\dot{S}(t) = -\beta S(t) I(t),
\ee
where
\be
\label{eq:DDE2}
I(t) = \varepsilon \mathbf{1}_{\{t<\tau\}} + S\big( (t-\tau)^+ \big) - S(t).
\ee
For simplicity, throughout the remainder of the paper we will assume the following delay initial conditions, which are compatible with \eqref{eq:ic_SIkR}:
\be
\label{eq:DDE3}
S(t) = 1 - \varepsilon, \quad \forall t \in [-\tau, 0].
\ee
\end{remark}

\begin{remark}\label{rem:2}
Note also that in view of \eqref{eq:VS},
\[
W(t) = 1 - S(t) - \frac{1}{R_0} \ln\bigg( \frac{1-\varepsilon}{S(t)} \bigg)
\]
and $W$ reaches its unique maximum when $S = \frac{1}{R_0}$, so that
\be
W_{\max} =\max_{t>0} W(t)= 1 - \frac{1}{R_0} - \frac{\ln \big( (1-\varepsilon) R_0 \big)}{R_0}.
\ee
\end{remark}
Defining
\[
I_{\max}=\max_{t>0} I(t),
\]
we may now ask how this quantity relates to $W_{\max}$. We address this question in the next section.

\subsection{The  approximation $I_{\max}^{(k)}\approx 2W_{\max}^{(k)}$}
\label{ssec:la}
We now use the result on infinite-stage limit established in Theorem~\ref{thm:1} to explain the approximation
\be
\label{eq:sc3}
I_{\max} \sim 2W_{\max}.
\ee
Indeed, Theorem~\ref{thm:1} shows that the solution $I(t)$ of the limiting delay system \eqref{eq:DDE1}--\eqref{eq:DDE2} admits a representation as an unweighted moving average of the incidence curve over the infectious period, whereas $W(t)$ corresponds to the same average computed with a triangular weighting kernel. More specifically,  
\begin{equation}\label{eq:int_rep}
I(t)=\int_{t-\tau}^{t}A(u)\,du,
\qquad
W(t)=\int_{t-\tau}^{t}
\left(1-\frac{t-u}{\tau}\right)A(u)\,du.
\end{equation}
This observation suggests comparing the two quantities through their local
behavior near the peak of the incidence curve. To this end, we may  apply a local
Laplace approximation \cite{deBruijn1981} around the incidence maximum. As it turns out, 
when the incidence profile is sufficiently broad near its peak, the leading-order
contribution to $I_{\max}$ is twice that of $W_{\max}$, yielding the
approximation~\eqref{eq:sc3}.

To show this, we first establish an important result concerning the relative timing of the peaks of $W$ and $I$. Let us define 
$$t_W = \argmax_{t>0} W(t)\quad\text{and}\quad t_I = \argmax_{t>0} I(t).$$

\begin{lemma}
\label{prop:2}
Assume $A(t)>0$ is strictly unimodal with unique maximizer $t^*$.
Then
\[
t_I-\tau<t^*\leq t_W<t_I.
\]
\end{lemma}
\noindent The proof of this lemma can be found in Appendix~\ref{app:pl3}. 

The integral representation \eqref{eq:int_rep} along with Lemma~\ref{prop:2} provides allows for more formal  insights into  the approximation~\eqref{eq:sc3}.  Denote

\begin{equation}\label{eq:ka}
\kappa
=
-\left.\frac{d^2}{dt^2}\ln A(t)\right|_{t=t^*}>0.
\end{equation} and \[
A_*=A(t^*).
\]
Then, by Lemma~\ref{prop:2}, near \(t^*\) we have
\[
\ln A(t^*+s)
=
\ln A_*
-\frac{\kappa s^2}{2}
+R(s),
\qquad
|R(s)|\le M|s|^3 ,
\]
for \(|s|\le \tau\). Setting \begin{equation}\label{eq:la}\lambda=\tau\sqrt{\kappa}\end{equation} and
\(t=t^*+x/\sqrt{\kappa}\), a standard Laplace expansion yields
\[
I(t)
=
\frac{A_*}{\sqrt{\kappa}}
\left[
\int_{x-\lambda}^{x}e^{-z^2/2}\,dz
+
O(\eta_I)
\right],
\text{ where }
\eta_I
=
\frac{M}{\kappa^{3/2}}
\sup_x
\int_{x-\lambda}^{x}|z|^3e^{-z^2/2}\,dz .
\]
Consequently,
\[
I_{\max}
=
\frac{A_*}{\sqrt{\kappa}}
\bigl[J_I(\lambda)+O(\eta_I)\bigr],
\]
with
\[
J_I(\lambda)
=
\max_x\int_{x-\lambda}^{x}e^{-z^2/2}\,dz
=
\int_{-\lambda/2}^{\lambda/2}e^{-z^2/2}\,dz
= 
\sqrt{2\pi}\left[2\Phi\left(\frac{\lambda}{2}\right)-1\right],
\]
where $\Phi$ is the standard normal CDF. Similarly,
\[
W_{\max}
=
\frac{A_*}{\sqrt{\kappa}}
\bigl[J_W(\lambda)+O(\eta_W)\bigr],
\]
where
\[
J_W(\lambda)
=
\max_x
\int_0^\lambda
\left(1-\frac{r}{\lambda}\right)
e^{-(x-r)^2/2}\,dr
\text{ and }
\eta_W
=
\frac{M}{\kappa^{3/2}}
\sup_x
\int_0^\lambda
\left(1-\frac{r}{\lambda}\right)
|x-r|^3e^{-(x-r)^2/2}\,dr .
\]

Therefore
\[
\frac{I_{\max}}{W_{\max}}
=
\frac{J_I(\lambda)+O(\eta_I)}
     {J_W(\lambda)+O(\eta_W)},
\]
and, provided \(J_W(\lambda)\) is bounded away from zero,
\[
I_{\max}
=
C(\lambda)W_{\max}
\left[
1+
O\!\left(
\frac{\eta_I}{J_I(\lambda)}
+
\frac{\eta_W}{J_W(\lambda)}
\right)
\right],
\]
where
\begin{equation*}
C(\lambda)
=
\frac{J_I(\lambda)}{J_W(\lambda)},
\end{equation*} and $\lambda$ is given by   \eqref{eq:la}.
The parameter $\lambda$ measures the length of the infectious period relative to the local width of the incidence peak. When $\lambda\ll1$, indicating that the incidence peak is broad relative to the infectious period, we obtain $C(\lambda) = 2 + O(\lambda^2)$ which yields~\eqref{eq:sc3}.

\subsection{Accuracy of the approximation}\label{ssec:accuracy}
The previous section provides  some formal justification  of the approximation $I_{\max} \approx 2W_{\max}$. We now analyze its accuracy by introducing the relative error
\[
E=\frac{2W_{\max}}{I_{\max}}-1.
\]
Below, we derive a theoretical condition for the approximation error to remain below a prescribed tolerance. Specifically, we show that a necessary condition for $E<\eta$, where $\eta>0$, is that the basic reproduction number $R_0$ lies below a threshold determined by $\eta$. We then support this result with numerical simulations showing that the relative error increases monotonically with $R_0$.

\begin{proposition}
\label{pro2}
Assume that $\varepsilon \to 0$ and let $\eta \in (0,1)$. If $E<\eta$, then $R_0<R_0^{\#}(\eta)$, where $R_0^{\#}(\eta)$ is the unique solution to
\be
\label{eq:Rs}
\frac{1+\ln(R_0)}{R_0} = \frac{1-\eta}{2}
\ee
on $R_0 \in (1, +\infty)$.
\end{proposition}
The proof is given in Appendix~\ref{app:accr}. This result provides a necessary condition: for the relative error to be smaller than a prescribed tolerance $\eta$, the basic reproduction number must satisfy $R_0<R_0^{\#}(\eta)$. For example, when $\eta=0.1$, we obtain $R_0^{\#}\approx 6.32$ (see Figure~\ref{fig:2}b). The following proposition provides a sharper necessary condition under which the approximation
\(I_{\max} \approx 2W_{\max} \)
is accurate within $\eta$-tolerance.
\begin{proposition}
\label{th:R1}
Assume that $\varepsilon \to 0$. Let $\eta \in (0,1)$ and $\psi \in (0,1/R_0)$. Suppose that $t_I>2\tau$ and that $t_I \notin \{\tau, 2\tau, .., j\tau,.. \}$. If $E<\eta$ and $S(t_I) - \frac{1}{R_0} > - \psi$, then
\be
\label{eq:R1}
1 - \frac{1+\ln(R_0)}{R_0} < \frac{\eta+1}{2} \cdot G\bigg(\frac{1}{R_0} - \psi \bigg),
\ee
where
\[
G(x) = \frac{1-x +\sqrt{1-2x+5x^2} }{2} -x.
\]
\end{proposition}

The proof is given in Appendix~\ref{app:accr}. This result provides a sharper necessary condition for the approximation $I_{\max}\approx 2W_{\max}$ to be accurate. Specifically, if the relative error $E$ is smaller than a prescribed tolerance $\eta$ and if the prevalence peak occurs near the susceptible level at which $W(t)$ is maximized, in the sense that $S(t_I)-\frac{1}{R_0}>-\psi$, then $R_0$ must satisfy~\eqref{eq:R1}. Thus, an accurate approximation is possible only within a restricted range of $R_0$ values determined jointly by the error tolerance $\eta$ and the peak-location tolerance $\psi$.

Numerically, \eqref{eq:R1} appears to define a threshold condition of the form $R_0<R_0^*(\eta,\psi)$ for $R_0<R_0^{\#}$. For example when $\eta=\psi=0.1$, this gives $R_0<3.86$. The next proposition gives an estimate for the upper bound on $\psi$.

\begin{proposition} 
\label{th:psi}
Assume that $R_0>1$, $t_I>2\tau$, and $t_I\notin\{\tau,2\tau,3\tau,\ldots\}$. Define
\be
\label{eg:a00}
a(r)= \frac{\ln(r)+\frac{r-1}{r} \big(\ln(r)-1+\frac{1}{r} \big)}{r-1}
\ee
and let $r^*(R_0)>1$ be the unique solution to
\be
\label{eg:a01}
a(r)\bigg(r+1-\frac{1}{r} \bigg)=R_0.
\ee
Then
\be
\label{eg:a02}
S(t_I)-\frac1{R_0} \geq \frac{a(r^*(R_0))-1}{R_0}.
\ee
\end{proposition}

\begin{figure}[ht!]
    \centering
    \subfigure[]{\includegraphics[width=0.49\linewidth]{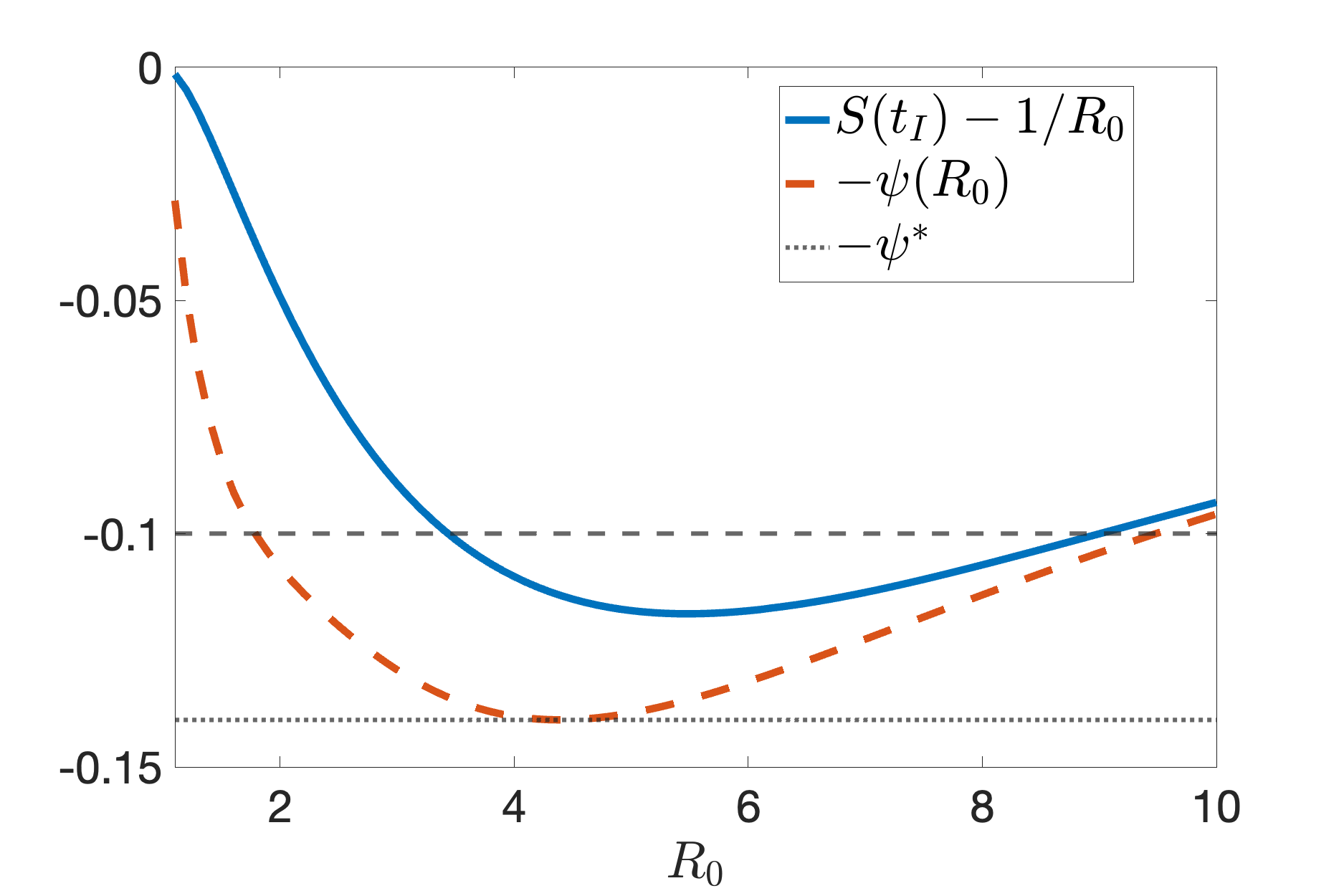}}
    \subfigure[]{\includegraphics[width=0.49\linewidth]{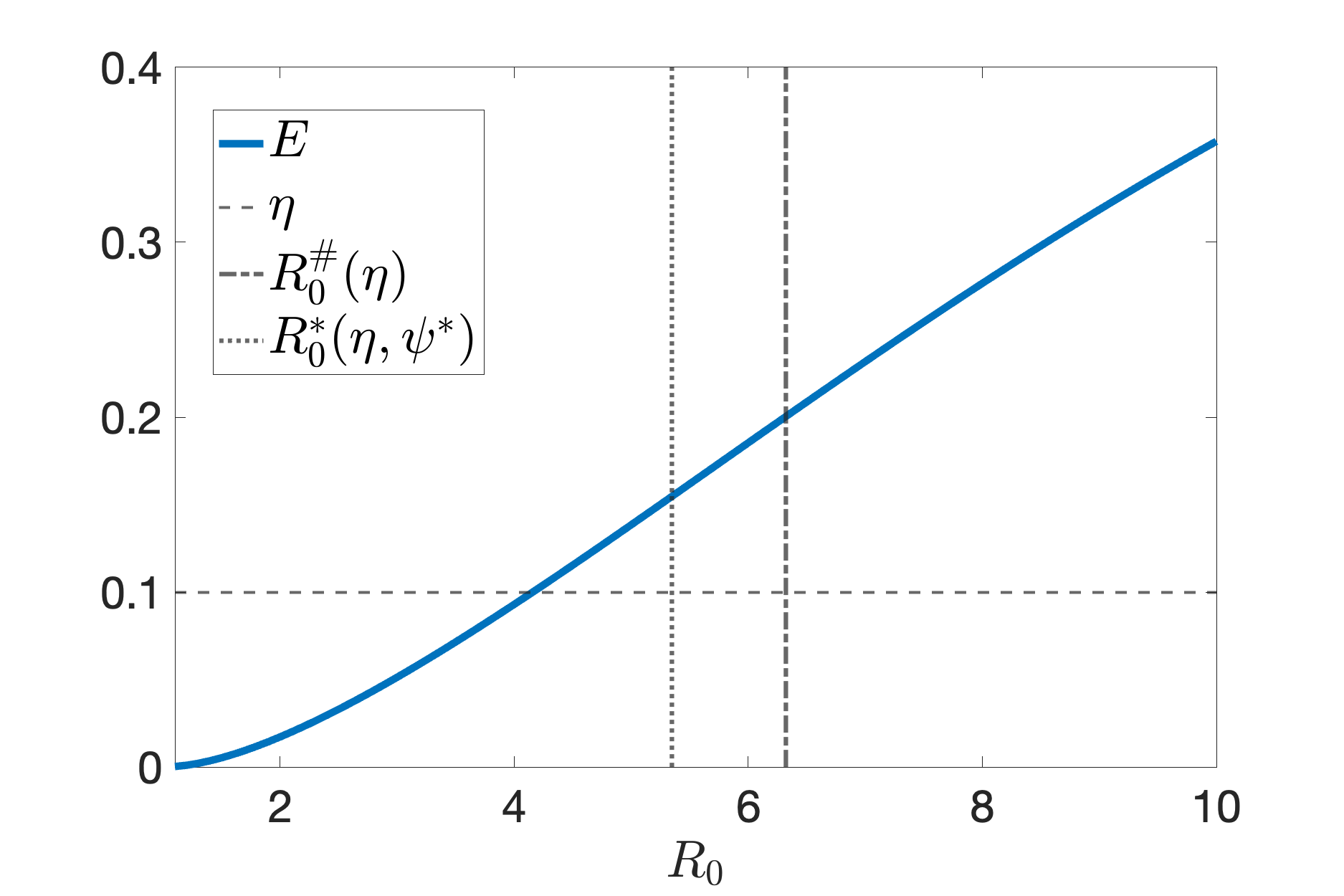}}
    \caption{
    Analytical and numerical constraints on the accuracy of the approximation $I_{\max}\approx 2W_{\max}$. (a) Theoretical lower bound $S(t_I)-1/R_0 \geq-\psi(R_0)$; the maximum over the relevant range is $\psi^*\approx 0.14$. (b) Analytical thresholds for error tolerance $\eta=0.1$: the basic necessary condition $E<\eta$ gives $R_0<R_0^{\#}(\eta)$, while the refined condition gives $R_0<R_0^*(\eta,\psi^*)$. Numerical simulations indicate that $E<0.1$ approximately for $1<R_0<4.2$. Parameters are $\varepsilon=10^{-4}$ and $\tau=10$.
    }
    \label{fig:2}
\end{figure}

The proof is given in Appendix~\ref{app:accr}. This result provides a parameter-dependent lower bound on how far the susceptible fraction at the prevalence peak can lie below the threshold $1/R_0$, $R_0 \in (1,R_0^{\#})$. Specifically, defining
\[
\psi(R_0)=\frac{1-a(r^*(R_0))}{R_0},
\]
we obtain
\[
S(t_I)-\frac{1}{R_0} \geq -\psi(R_0).
\]
This lower bound is illustrated in Figure~\ref{fig:2}a. Since $a(r)\in(0,1)$ for all $r>1$, it follows that for each $R_0\in(1,R_0^{\#})$, $\psi(R_0) \in (0,1/R_0)$. Therefore, Proposition~\ref{th:R1} can be applied with $\psi=\psi(R_0)$. If $E<\eta$, then
\[
1 - \frac{1+\ln(R_0)}{R_0} < \frac{\eta+1}{2} \cdot G\bigg( \frac{a(r^*(R_0))}{R_0} \bigg).
\]
Now define
\[
\psi^* = \sup_{R_0 \in (1,R_0^{\#})} \psi(R_0).
\]
Numerical evaluation gives $\psi^* \approx 0.139 < \frac{1}{R_0^{\#}} = 0.158$. Hence, for every $R_0 \in (1,R_0^{\#})$, $\psi^*<\frac{1}{R_0^{\#}} < \frac{1}{R_0}$. Thus Proposition~\ref{th:R1} can also be applied with the uniform tolerance $\psi^*$ yielding the necessary condition $R_0 < R^*(\eta,\psi^*)$. For $\eta=0.1$, this gives approximately $R_0<5.35$.

The thresholds $R_0^\#$ and $R^*(0.1,\psi^*)$ are shown in Figure~\ref{fig:2}b. For comparison, numerical simulations indicate that the relative error $E$ increases monotonically with $R_0$ and remains below $0.1$ approximately for
\[
1<R_0<4.2.
\]
These results clarify the applicability range of the approximation $I_{\max}\approx 2W_{\max}$. The infinite-stage limit explains the origin of the factor of two, but the analytical thresholds and numerical simulations show that the approximation is reliable only over a restricted range of $R_0$. For larger $R_0$, the relative error becomes substantial, motivating refined approximations that incorporate higher-order information about the local geometry of the incidence curve near its peak.

\subsection{Refined approximations}\label{ssec:refined}
The previous section showed that the simple approximation $I_{\max}\approx 2W_{\max}$ is accurate only over a restricted range of $R_0$, with the relative error increasing substantially as $R_0$ grows. This motivates the construction of refined approximations that retain the interpretability of the simple formula while accounting for higher-order features of the incidence curve. The infinite-stage limit established in Section~\ref{ssec:lim} and the local Laplace approximation developed in Section~\ref{ssec:la} provide a natural basis for such refinements.

Recall that the local Laplace approximation gives
\[
I_{\max} = C(\lambda)W_{\max} \left[1+O\!\left(\frac{\eta_I}{J_I(\lambda)} + \frac{\eta_W}{J_W(\lambda)}\right)\right],
\]
where
\[
C(\lambda) = \frac{J_I(\lambda)}{J_W(\lambda)}, \qquad \lambda=\tau\sqrt{\kappa},
\] and $\kappa$ is given by $\eqref{eq:ka}$
This representation allows us to derive the following family of approximations:

\begin{itemize}
\item
\textbf{The fully corrected approximation (FC).}
\begin{equation}\label{eq:FC_est}
\widehat I_{\max}^{\mathrm{FC}}(\lambda) = C(\lambda)W_{\max}.
\end{equation}

\item
\textbf{The first-order corrected approximation (FO).} Assuming \(\lambda\ll1\), $C(\lambda) = 2\left(1-\frac{\lambda^2}{72}\right) + O(\lambda^4)$ (see Appendix~\ref{app:FO_derivation} for the derivation) so that
\begin{equation}\label{eq:FO_est}
\widehat I_{\max}^{\mathrm{FO}}(\lambda) = 2W_{\max} \left(1-\frac{\lambda^2}{72}\right).
\end{equation}

\item
\textbf{The simple approximation (S).} Neglecting the curvature correction in the previous approximation yields
\begin{equation}\label{eq:S_est}
\widehat I_{\max}^{\mathrm{S}} = 2W_{\max}.
\end{equation}

\item
\textbf{Large-$\lambda$ zero- and one-step corrected approximations $L(0)$ and $L(1)$.} Assuming instead that $\lambda > \sqrt{2\pi}$, we get (see Appendix~\ref{app:L_derivation} for the derivation)
\be
\label{eq:L_est}
\widehat I_{\max}^{L(m)}(\lambda) = \Bigg[\frac{\sqrt{2\pi}\left[2\Phi(\lambda/2)-1\right]}{e^{-(x_W^{(m)})^2/2}\left(\lambda-x_W^{(m)}-\lambda^{-1}\right)}\Bigg] W_{\max},
\ee
where
\[
x_W^{(0)} = \left[ 2\ln\left(\frac{\lambda}{\sqrt{2\pi}}\right)\right]^{1/2}
\]
and
\[
x_W^{(1)} = \left[2\ln\left(\frac{\lambda}{\sqrt{2\pi}\Phi(x_W^{(0)})}\right)\right]^{1/2}.
\]
\end{itemize}

We can now compare these approximations of $I_{\max}$. Relative errors are reported as
\[
\mathrm{E}^{X} = \frac{\widehat I_{\max}^{X}-I_{\max}}{I_{\max}}, \qquad
X\in\{\mathrm{FC},\mathrm{FO},\mathrm{S},\mathrm{L(0)},\mathrm{L(1)}\}.
\]

\begin{table}[ht!]
\centering
\caption{
Comparison of errors for the simple (S), first-order corrected (FO), fully corrected (FC), and large-$\lambda$ zero- and one-step corrected approximations, denoted $L(0)$ and $L(1)$, respectively. Blue bolded entries indicate errors with absolute value less than $0.1$. Parameters are $\tau=10$ and $\varepsilon=10^{-4}$.
}
\label{tab:Imax_errors}
\begin{tabular}{c c c c c c c c}
\hline
$R_0$ & $\kappa$ & $\lambda$ & $E^S$ & $E^{FO}$ & $E^{FC}$ & $E^{L(0)}$ & $E^{L(1)}$ \\
\hline
$1.5$  & $0.003705$ & $0.609$ & $\bf\textcolor{blue}{0.005}$ & $\bf\textcolor{blue}{0.000}$  & $\bf\textcolor{blue}{0.000}$  & --       & --       \\
$2.0$  & $0.012303$ & $1.109$ & $\bf\textcolor{blue}{0.017}$ & $\bf\textcolor{blue}{-0.000}$ & $\bf\textcolor{blue}{0.000}$  & --       & --       \\
$2.5$  & $0.024121$ & $1.553$ & $\bf\textcolor{blue}{0.033}$ & $\bf\textcolor{blue}{-0.002}$ & $\bf\textcolor{blue}{-0.000}$ & --       & --       \\
$3.0$  & $0.038555$ & $1.964$ & $\bf\textcolor{blue}{0.051}$ & $\bf\textcolor{blue}{-0.005}$ & $\bf\textcolor{blue}{-0.001}$ & --       & --       \\
$3.5$  & $0.055364$ & $2.353$ & $\bf\textcolor{blue}{0.072}$ & $\bf\textcolor{blue}{-0.011}$ & $\bf\textcolor{blue}{-0.003}$ & --       & --       \\
$4.0$  & $0.074451$ & $2.729$ & $\bf\textcolor{blue}{0.093}$ & $\bf\textcolor{blue}{-0.020}$ & $\bf\textcolor{blue}{-0.006}$ & $-0.367$ & $0.374$  \\
$4.5$  & $0.095758$ & $3.094$ & $0.115$ & $\bf\textcolor{blue}{-0.033}$ & $\bf\textcolor{blue}{-0.009}$ & $-0.286$ & $0.159$  \\
$5.0$  & $0.119413$ & $3.456$ & $0.138$ & $\bf\textcolor{blue}{-0.050}$ & $\bf\textcolor{blue}{-0.013}$ & $-0.238$ & $\bf\textcolor{blue}{0.083}$  \\
$5.5$  & $0.145288$ & $3.812$ & $0.162$ & $\bf\textcolor{blue}{-0.073}$ & $\bf\textcolor{blue}{-0.018}$ & $-0.207$ & $\bf\textcolor{blue}{0.045}$  \\
$6.0$  & $0.173470$ & $4.165$ & $0.185$ & $-0.100$ & $\bf\textcolor{blue}{-0.023}$ & $-0.185$ & $\bf\textcolor{blue}{0.021}$  \\
$6.5$  & $0.203985$ & $4.516$ & $0.209$ & $-0.134$ & $\bf\textcolor{blue}{-0.028}$ & $-0.170$ & $\bf\textcolor{blue}{0.005}$  \\
$7.0$  & $0.236858$ & $4.867$ & $0.232$ & $-0.173$ & $\bf\textcolor{blue}{-0.033}$ & $-0.159$ & $\bf\textcolor{blue}{-0.008}$ \\
$7.5$  & $0.272096$ & $5.216$ & $0.254$ & $-0.220$ & $\bf\textcolor{blue}{-0.037}$ & $-0.150$ & $\bf\textcolor{blue}{-0.017}$ \\
$8.0$  & $0.309712$ & $5.565$ & $0.277$ & $-0.273$ & $\bf\textcolor{blue}{-0.040}$ & $-0.143$ & $\bf\textcolor{blue}{-0.025}$ \\
$8.5$  & $0.349480$ & $5.912$ & $0.298$ & $-0.332$ & $\bf\textcolor{blue}{-0.043}$ & $-0.137$ & $\bf\textcolor{blue}{-0.030}$ \\
$9.0$  & $0.392103$ & $6.262$ & $0.319$ & $-0.399$ & $\bf\textcolor{blue}{-0.046}$ & $-0.132$ & $\bf\textcolor{blue}{-0.035}$ \\
$9.5$  & $0.437260$ & $6.613$ & $0.339$ & $-0.474$ & $\bf\textcolor{blue}{-0.048}$ & $-0.128$ & $\bf\textcolor{blue}{-0.038}$ \\
$10.0$ & $0.484408$ & $6.960$ & $0.358$ & $-0.556$ & $\bf\textcolor{blue}{-0.049}$ & $-0.123$ & $\bf\textcolor{blue}{-0.041}$ \\
\hline
\end{tabular}
\end{table}

\begin{figure}[ht!]
    \centering
    \subfigure[]{\includegraphics[width=0.495\linewidth]{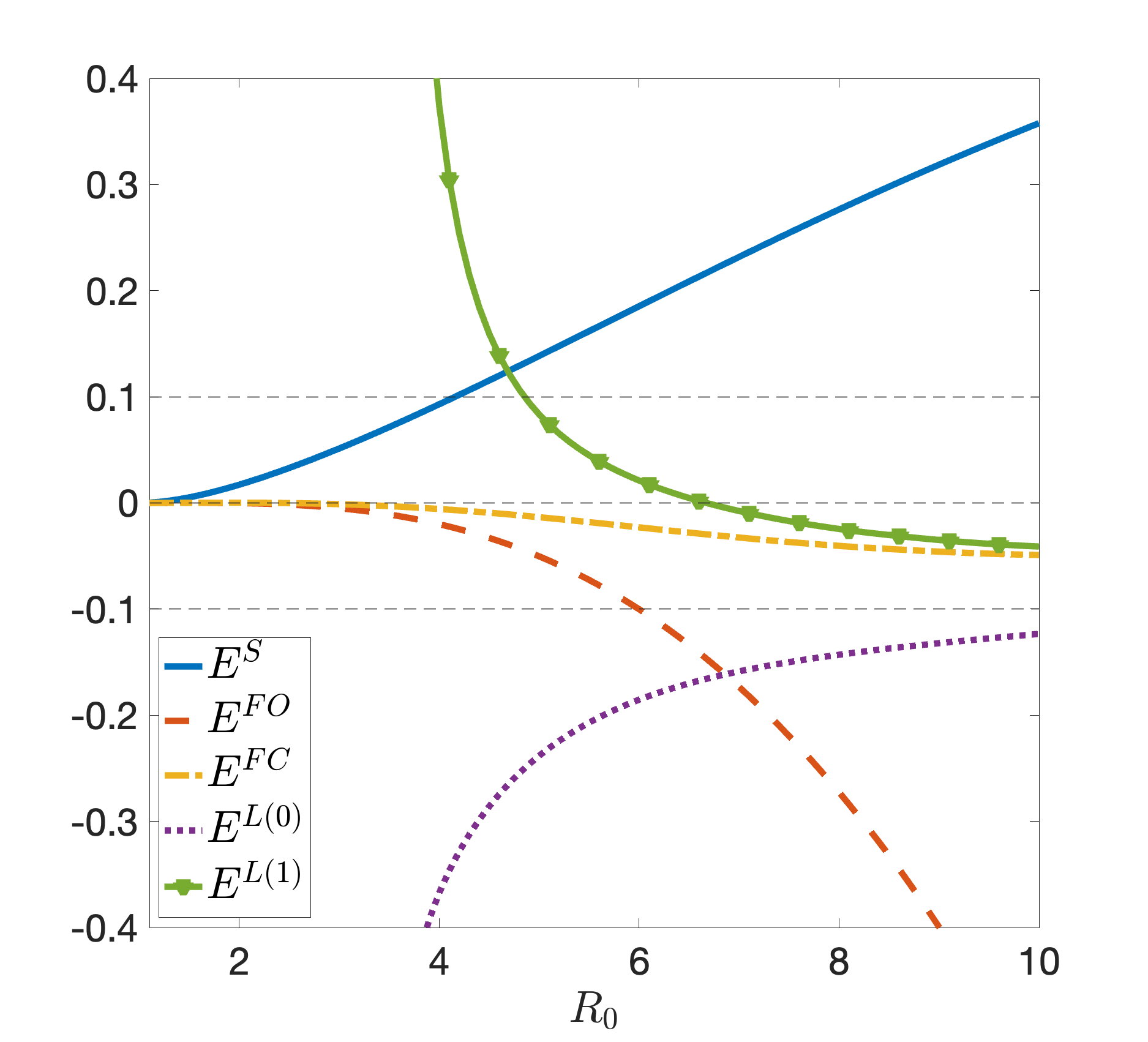}}
    \subfigure[]{\includegraphics[width=0.495\linewidth]{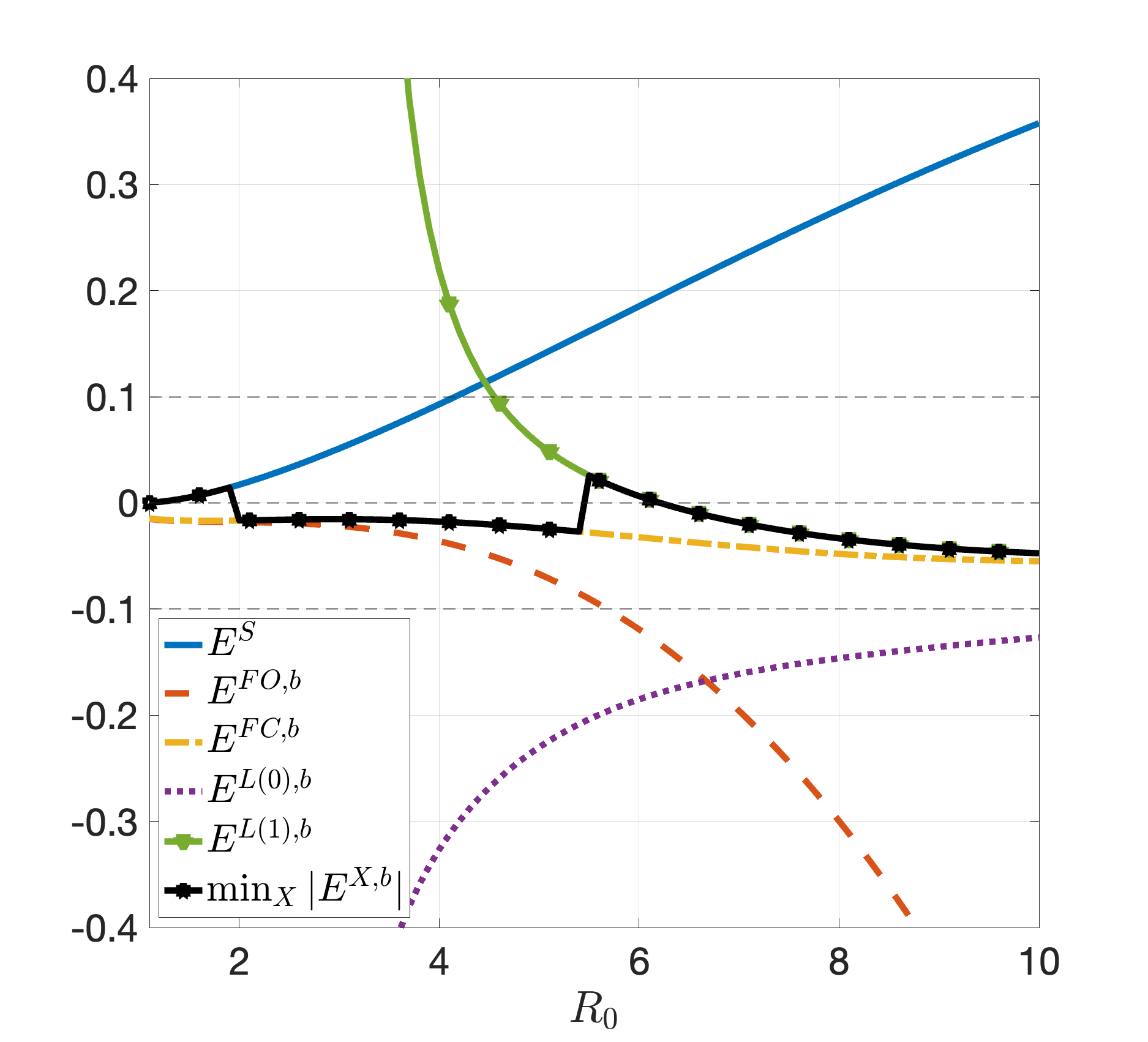}}
    \caption{
    Relative errors $E$ of different approximations of $I_{\max}$ as functions of $R_0$. (a) Errors for the simple (S), first-order corrected (FO), fully corrected (FC), and large-$\lambda$ corrected approximations $L(0)$ and $L(1)$, using trajectory-based $\lambda$. (b) Corresponding plug-in errors using $\lambda=\sqrt{(R_0^2+1)/2}$; the black line marks the approximation with the smallest absolute error for each $R_0$. In both panels, $E=\widehat I_{\max}/I_{\max}-1$, with positive values indicating overestimation and negative values indicating underestimation. Horizontal dashed lines mark $E=0$ and the $\pm 0.1$ thresholds. Parameters are $\tau=10$ and $\varepsilon=10^{-4}$.
    }
    \label{fig:3}
\end{figure}

Table~\ref{tab:Imax_errors} and Figure~\ref{fig:3}a show that the first-order correction (FO) extends the range over which the simple approximation is accurate, remaining below the $0.1$ error threshold up to approximately $R_0=6.0$, but it eventually overcorrects and substantially underestimates $I_{\max}$ for larger $R_0$, consistent with the fact that FO is based on a small-$\lambda$ expansion. The fully corrected estimator (FC) performs best across the full range of $R_0$ considered. The large-$\lambda$ zero-step approximation $L(0)$ is too crude and systematically underestimates $I_{\max}$. In contrast, the one-step approximation $L(1)$ is substantially more accurate: after an initial failure near $R_0=4$, it remains within the $0.1$ error threshold for $R_0\geq5$. 

Overall, these results indicate a clear hierarchy. The simple estimator is accurate only when the incidence peak is broad, corresponding to smaller $R_0$. FO improves the approximation for intermediate $R_0$, while $L(1)$ provides a useful explicit approximation for larger $R_0$. FC remains accurate across the full range considered. A limitation of the refined approximations is their dependence on the parameter $\kappa$, which is obtained from a portion of the epidemic trajectory. We therefore next ask whether comparable approximations can be constructed using only model parameters, without requiring trajectory-level information. To this end, we first establish the following result.

\begin{lemma}
\label{th:lam}
Assume that $R_0 > 1$ and the incidence function $A(t)$ is strictly unimodal with unique maximizer $t^*>\tau$. Then
\[
\lambda\leq \sqrt{\frac{R_0^2+1}{2}}<R_0.
\]
\end{lemma}

The proof is given in Appendix~\ref{app:plam}. The lemma provides an explicit upper bound for the width parameter $\lambda$. We use this bound as a plug-in value in the refined approximations in order to obtain formulas that depend only on $R_0$. Specifically, define the bound-based plug-in approximations $\widehat I_{\max}^{\mathrm{FC}} \bigg( \sqrt{\frac{R_0^2+1}{2}} \bigg)$, $\widehat I_{\max}^{\mathrm{FO}} \bigg( \sqrt{\frac{R_0^2+1}{2}} \bigg)$, and $\widehat I_{\max}^{L(m)} \bigg( \sqrt{\frac{R_0^2+1}{2}} \bigg)$. The results on their performances are shown in Table~\ref{tab:Imax_errors1} and Figure~\ref{fig:3}b, indicating that the simple approximation performs best for $1<R_0<2$, the fully corrected plug-in approximation performs best for $2\leq R_0<5.5$, and the large-$\lambda$ one-step corrected plug-in approximation performs best for $R_0\geq 5.5$. Overall, the combined approximation
\[
I_{\max} \approx 
\begin{cases}
\widehat I_{\max}^{\mathrm{S}}, \text{ if } R_0 \in (1,2),\\
\widehat I_{\max}^{\mathrm{FC}} \bigg( \sqrt{\frac{R_0^2+1}{2}} \bigg), \text{ if } R_0 \in [2,5.5),\\
\widehat I_{\max}^{L(1)} \bigg( \sqrt{\frac{R_0^2+1}{2}} \bigg), \text{ if } R_0 \in [5.5,10].\\
\end{cases}
\]
has absolute relative error $|E|<0.05$ over the parameter range considered. The performance of the bound-based plug-in approximations for finite values of $k$ is shown in Figures~\ref{fig:5} and~\ref{fig:6}.

\begin{table}[ht!]
\centering
\caption{
Comparison of errors for the simple $(S)$, first-order corrected plug-in $(FO,p)$, fully corrected plug-in $(FC,p)$, and large-$\lambda$ zero- and one-step corrected plug-in approximations, denoted $(L(0),p)$ and $(L(1),p)$, respectively. The plug-in approximations use the bound-based value $\lambda_b=\sqrt{(R_0^2+1)/2}$. Blue bolded  entries indicate for each value of $R_0$ the approximation with the smallest absolute error among the listed approximations. Parameters are $\tau=10$ and $\varepsilon=10^{-4}$.
}
\label{tab:Imax_errors1}
\begin{tabular}{c c c c c c c c}
\hline
$R_0$ & $\kappa$ & $\lambda_b$ & $E^S$ & $E^{FO}_p$ & $E^{FC}_p$ & $E^{L(0)}_p$ & $E^{L(1)}_p$ \\
\hline
$1.5$ & $0.003705$ & $1.275$ & $\bf\textcolor{blue}{0.005}$ & $-0.017$ & $-0.017$ & -- & -- \\
$2.0$ & $0.012303$ & $1.581$ & $0.017$ & $-0.018$ & $\bf\textcolor{blue}{-0.017}$ & -- & -- \\
$2.5$ & $0.024121$ & $1.904$ & $0.033$ & $-0.019$ & $\bf\textcolor{blue}{-0.016}$ & -- & -- \\
$3.0$ & $0.038555$ & $2.236$ & $0.051$ & $-0.022$ & $\bf\textcolor{blue}{-0.015}$ & -- & -- \\
$3.5$ & $0.055364$ & $2.574$ & $0.072$ & $-0.027$ & $\bf\textcolor{blue}{-0.016}$ & $-0.434$ & $0.649$ \\
$4.0$ & $0.074451$ & $2.915$ & $0.093$ & $-0.036$ & $\bf\textcolor{blue}{-0.018}$ & $-0.326$ & $0.219$ \\
$4.5$ & $0.095758$ & $3.260$ & $0.115$ & $-0.049$ & $\bf\textcolor{blue}{-0.020}$ & $-0.268$ & $0.107$ \\
$5.0$ & $0.119413$ & $3.606$ & $0.138$ & $-0.067$ & $\bf\textcolor{blue}{-0.024}$ & $-0.230$ & $0.055$ \\
$5.5$ & $0.145288$ & $3.953$ & $0.162$ & $-0.090$ & $-0.028$ & $-0.204$ & $\bf\textcolor{blue}{0.026}$ \\
$6.0$ & $0.173470$ & $4.301$ & $0.185$ & $-0.119$ & $-0.032$ & $-0.185$ & $\bf\textcolor{blue}{0.006}$ \\
$6.5$ & $0.203985$ & $4.650$ & $0.209$ & $-0.154$ & $-0.037$ & $-0.171$ & $\bf\textcolor{blue}{-0.008}$ \\
$7.0$ & $0.236858$ & $5.000$ & $0.232$ & $-0.196$ & $-0.041$ & $-0.161$ & $\bf\textcolor{blue}{-0.018}$ \\
$7.5$ & $0.272096$ & $5.350$ & $0.254$ & $-0.244$ & $-0.045$ & $-0.153$ & $\bf\textcolor{blue}{-0.027}$ \\
$8.0$ & $0.309712$ & $5.701$ & $0.277$ & $-0.300$ & $-0.048$ & $-0.146$ & $\bf\textcolor{blue}{-0.033}$ \\
$8.5$ & $0.349480$ & $6.052$ & $0.298$ & $-0.362$ & $-0.051$ & $-0.141$ & $\bf\textcolor{blue}{-0.038}$ \\
$9.0$ & $0.392103$ & $6.403$ & $0.319$ & $-0.432$ & $-0.053$ & $-0.136$ & $\bf\textcolor{blue}{-0.042}$ \\
$9.5$ & $0.437260$ & $6.755$ & $0.339$ & $-0.510$ & $-0.054$ & $-0.131$ & $\bf\textcolor{blue}{-0.045}$ \\
$10.0$ & $0.484408$ & $7.106$ & $0.358$ & $-0.595$ & $-0.055$ & $-0.127$ & $\bf\textcolor{blue}{-0.047}$ \\
\hline
\end{tabular}
\end{table}

\begin{figure}[ht!]
    \centering
    \includegraphics[width=1\linewidth]{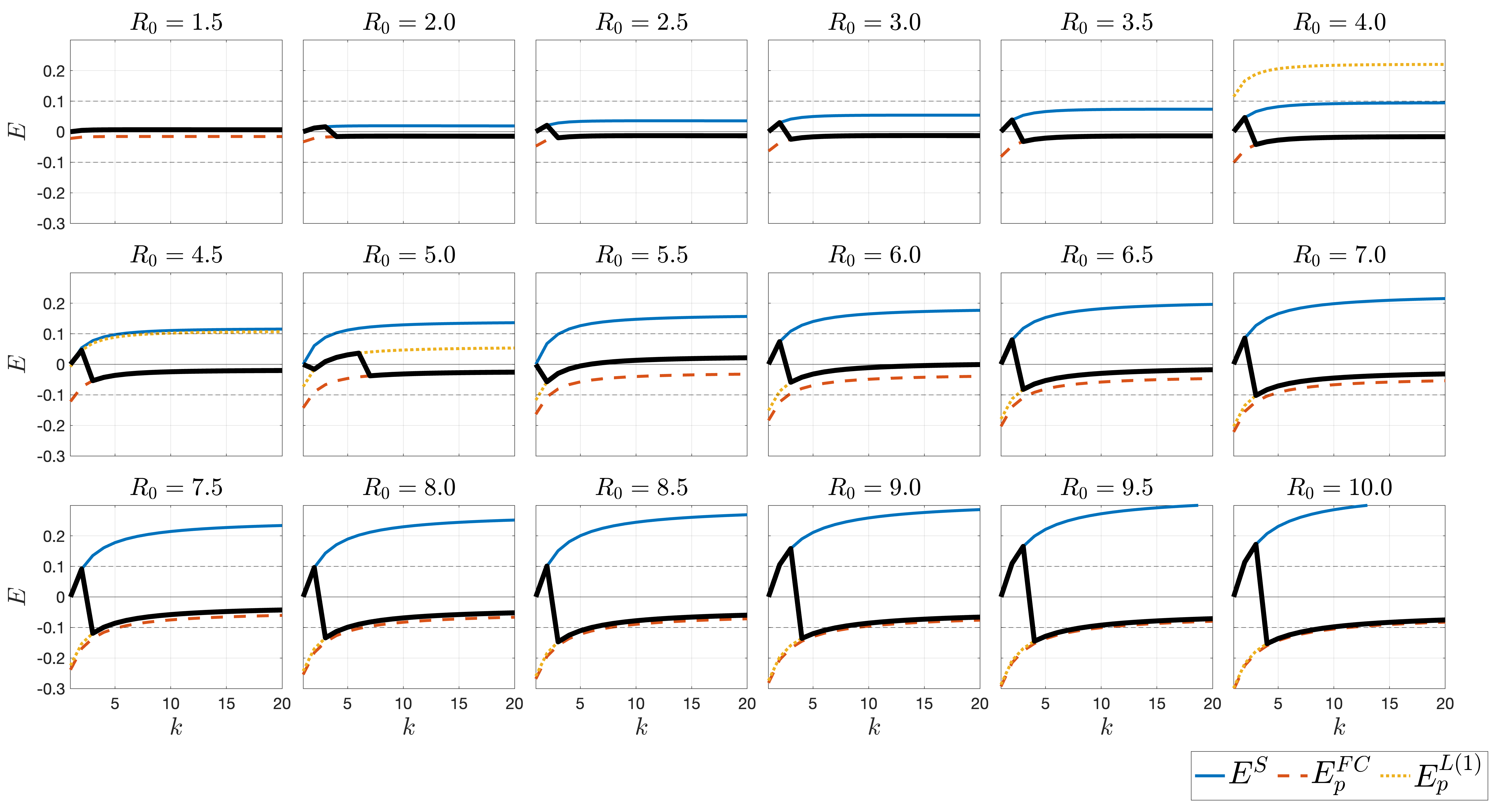}
    \caption{
    Relative errors of approximations of $I_{\max}^{(k)}$ in the finite $SI(k)R$ model as functions of the number of infectious stages $k$, shown for different values of $R_0$. The figure compares the simple approximation $(S)$, the fully corrected plug-in approximation $(FC,p)$ using $\lambda=\sqrt{(R_0^2+1)/2}$, and the corresponding large-$\lambda$ one-step approximation ($L(1),p$). The scaled weighted peak is $W_{\max}=V_{\max}^{(k)}/(k+1)$, and the error is $E=\widehat I_{\max}^{(k)}/I_{\max}^{(k)}-1$, with positive values indicating overestimation. Black lines indicate the approximation with the smallest absolute error for each $k$. Horizontal dashed lines mark $E=0$ and the $\pm0.1$ thresholds. Parameters are $\tau=10$ and $\varepsilon=0.01$.
    }
    \label{fig:5}
\end{figure}

\section{Summary and Discussion}\label{sec:summary}

Estimating the maximum prevalence of infection is a fundamental problem in epidemic modeling because peak prevalence is closely related to healthcare demand, workforce disruption, and overall epidemic burden. In practice, however, prevalence is often difficult to observe directly, whereas incidence-related quantities are more readily available from surveillance systems. Motivated by this observation, we investigated the relationship between peak prevalence and the maximum of a weighted stage aggregate in multistage SIR epidemic models.

Our analysis reveals that the behavior of this relationship depends critically on how the infectious-stage structure is scaled as the number of stages increases. Under the naive scaling, in which the stage progression rate remains fixed while the number of stages grows, the epidemic occupies only a small portion of the infectious chain. In this regime, the weighted stage aggregate becomes asymptotically equivalent to prevalence itself, and the approximation \(I_{\max}\approx 2W_{\max}\) fails. Instead, the weighted and unweighted peaks converge to the same limiting quantity.

The situation is fundamentally different under Erlang scaling, where the stage progression rate increases proportionally with the number of stages and the mean infectious period remains fixed. In this case, we establish a infinite-stage limit that connects the finite-dimensional multistage model to a delay epidemic formulation. The limiting representation provides a transparent interpretation of the two key quantities studied in the paper: prevalence becomes a moving average of incidence over the infectious period, whereas the weighted stage functional becomes a triangularly weighted moving average over the same interval. This observation explains why the approximation
\[
I_{\max}\approx 2W_{\max}
\]
arises naturally in the limit of a large number of infectious stages under Erlang scaling. 

A second contribution of the paper is to quantify the accuracy of this approximation. Our theoretical results show that the quality of the approximation depends strongly on the epidemic growth potential, as measured by the basic reproduction number \(R_0\). Both analytical bounds and numerical experiments indicate that the approximation is highly accurate for low values of \(R_0\), but becomes progressively less accurate as epidemic waves become sharper and more concentrated. In particular, numerical results suggest that the simple factor-two approximation remains within approximately ten percent relative error for a substantial range of epidemiologically relevant parameter values.

To improve upon the basic approximation, we derived a family of asymptotic corrections based on the local geometry of the incidence curve near its peak. These corrections show that the discrepancy between \(I_{\max}\) and \(2W_{\max}\) is controlled by the characteristic width of the incidence peak and can be expressed through a curvature parameter associated with the incidence trajectory. The resulting corrected approximations remain highly accurate over a much broader range of transmission intensities than the simple factor-two rule.

An important practical observation is that the curvature parameter is generally not known a priori and may be difficult to estimate reliably from epidemic data. To address this issue, we derived explicit parameter-based approximations that eliminate the need to estimate epidemic curvature directly. In particular, the large-\(\lambda\) one-step  plug-in correction \(L(1)\) depends only on model parameters and provides accurate estimates of peak prevalence when \(R_0\) is relatively large. Numerical experiments suggest that combining the simple approximation, the corrected plug-in approximation, and the \(L(1)\) plug-in approximation yields uniformly good performance across a broad range of epidemic scenarios, with relative errors typically below five percent over the parameter range examined.

Beyond the specific approximations developed here, the analysis highlights a broader connection between multistage epidemic models and delay epidemic formulations. The infinite-stage limit provides a useful framework for understanding how quantities defined in finite-dimensional compartmental models translate into weighted incidence functionals in limiting age-of-infection descriptions. This perspective may prove useful in other settings where direct observation of prevalence is difficult but incidence data are readily available.

Several extensions remain possible. The present work focuses on deterministic epidemic dynamics with homogeneous mixing and Erlang-distributed infectious periods. It would be of interest to investigate analogous prevalence approximations in stochastic epidemic models, heterogeneous contact structures, and more general infectious-period distributions. Another natural direction is the development of data-driven methods for estimating the weighted incidence functionals studied here from partially observed surveillance data.

In summary, the paper provides a rigorous analysis of the relationship between peak prevalence and weighted incidence in multistage epidemic models. We show that the widely observed approximation \(I_{\max}\approx 2W_{\max}\) emerges naturally under Erlang scaling, identify the conditions under which it is accurate, derive higher-order corrections that substantially improve its performance, and develop practical parameter-only approximations that remain effective across a wide range of epidemic regimes. These results provide both theoretical insight and useful tools for estimating epidemic burden from quantities that are often easier to observe or reconstruct in practice.

\newpage
\backmatter

\section*{Statements and Declarations}

\subsection*{Funding}

The authors graciously acknowledge the support   through  OSU HELM Initiative. 

\subsection*{Competing interests}

The authors have no relevant financial or non-financial interests to disclose.

\subsection*{Data availability}

No new datasets were generated or analyzed during the preparation of this article. The code used to produce the analytical and numerical results is publicly available in the GitHub repository: \href{https://github.com/dtverskoi/Approximating-Peak-Prevalence-in-Multistage-SIR-Epidemics}{https://github.com/dtverskoi/Approximating-Peak-Prevalence-in-Multistage-SIR-Epidemics}

\subsection*{Author contributions}

DT and GAR conceived and designed the study, derived the analytical results, and performed the numerical simulations. AG contributed to the numerical simulations and developed the software used for the analyses. All authors read and approved the final manuscript.
\clearpage

\phantomsection
\addcontentsline{toc}{section}{References}

\newpage
\begin{appendices}

\setcounter{subsection}{0}
\setcounter{equation}{0}
\setcounter{figure}{0}
\setcounter{table}{0}

\renewcommand{\thesubsection}{A\arabic{subsection}}
\renewcommand{\theequation}{A\arabic{equation}}
\renewcommand{\thefigure}{A\arabic{figure}}
\renewcommand{\thetable}{A\arabic{table}}

\section*{Appendix}
\addcontentsline{toc}{section}{Appendix}

\subsection{Point mass limit in Theorem~\ref{thm:1}}
\label{app:wc}

To complete the proof of Theorem~\ref{thm:1}, it remains to justify the weak limit of the initial cohort.  For fixed \(t>0\),
consider the measure
\[
\mu_k^t(dx):=\varepsilon k p_{k,\lfloor kx\rfloor}(t)\,dx .
\]
For any bounded continuous test function \(\varphi\),
\[
\int_0^1 \varphi(x)\,\mu_k^t(dx)
=
\varepsilon\int_0^1 \varphi(x) k p_{k,\lfloor kx\rfloor}(t)\,dx .
\]
Partitioning \([0,1]\) into intervals \([i/k,(i+1)/k)\), we obtain
\[
\int_0^1 \varphi(x)\,\mu_k^t(dx)
=
\varepsilon
\sum_{i=0}^{k-1} p_{k,i}(t)
\int_{i/k}^{(i+1)/k} k\varphi(x)\,dx .
\]
By uniform continuity of \(\varphi\),
\[
\int_{i/k}^{(i+1)/k} k\varphi(x)\,dx
=
\varphi(i/k)+o(1),
\]
uniformly in \(i\). Hence
\[
\int_0^1 \varphi(x)\,\mu_k^t(dx)
=
\varepsilon
\sum_{i=0}^{k-1}\varphi(i/k)p_{k,i}(t)+o(1).
\]

Now \(p_{k,i}(t)\) is the probability that the initial individual is in stage
\(i\) at time \(t\). Equivalently, if \(N_k(t)\) is a Poisson process with rate
\(\frac{k}{\tau}\), then $p_{k,i}(t)=\mathbb P\{N_k(t)=i-1\}$ up to the harmless indexing convention. Therefore
\[
\sum_{i=0}^{k-1}\varphi(i/k)p_{k,i}(t)
=
\mathbb E\left[\varphi\left(\frac{N_k(t)}{k}\right)\right]+o(1).
\]
Since $\frac{N_k(t)}{k}\to \frac{t}{\tau}$ in probability by the law of large numbers for the Poisson process, bounded convergence gives $\mathbb E\left[\varphi\left(\frac{N_k(t)}{k}\right)\right]
\to
\varphi \big( \frac{t}{\tau} \big).$ Consequently,
\[
\int_0^1 \varphi(x)\,\mu_k^t(dx)
\to
\varepsilon \varphi\bigg( \frac{t}{\tau} \bigg).
\]
This proves the weak convergence
\[
\mu_k^t(dx)
=
\varepsilon k p_{k,\lfloor kx\rfloor}(t)\,dx
\Longrightarrow
\varepsilon\delta_{t / \tau}(dx).
\]

\subsection{Proof of Lemma 3}
\label{app:pl3}
\begin{proof}
Note that
\[
A(t_W) = \beta S(t_W)I(t_W) = \frac{1}{\tau}I(t_W).
\]
Using $I(t_W)=\int_{t_W-\tau}^{t_W}A(u)du$, we obtain
\be
\label{eq:Atv}
A(t_W) = \frac{1}{\tau}\int_{t_W-\tau}^{t_W}A(u)du,
\ee
meaning that $A(t_W)$ equals the average value of $A$ over $[t_W-\tau,t_W]$.

We first show  \(t^*\le t_W\). Indeed, if  \(t_W<t^*\), then \(A\) would be  strictly increasing on \([t_W-\tau,t_W]\), implying
\[
A(t_W)>
\frac1\tau\int_{t_W-\tau}^{t_W}A(u)\,du,
\]
contradicting \eqref{eq:Atv}. Hence we must have \(t^*\le t_W\).

Similarly, \(t^*\) cannot satisfy \(t^*\le t_W-\tau\). Otherwise \(A\) is
strictly decreasing on \([t_W-\tau,t_W]\), yielding
\[
A(t_W)<
\frac1\tau\int_{t_W-\tau}^{t_W}A(u)\,du,
\]
again contradicting \eqref{eq:Atv}. Therefore
\begin{equation}\label{eq:1}
t_W-\tau<t^*\le t_W.
\end{equation}

Next we show that \(A(t_W)>A(t_W-\tau)\). Suppose instead that
\(A(t_W)\le A(t_W-\tau)\). By strict unimodality and \eqref{eq:1},
\[
A(u)\geq A(t_W),
\qquad
u\in[t_W-\tau,t_W],
\]
with strict inequality on a set of positive measure. Consequently,
\[
\frac1\tau\int_{t_W-\tau}^{t_W}A(u)\,du>A(t_W),
\]
contradicting \eqref{eq:Atv}. Thus
\[
A(t_W)>A(t_W-\tau).
\]
Since \(I'(t)=A(t)-A(t-\tau)\), it follows that
\[
I'(t_W)=A(t_W)-A(t_W-\tau)>0,
\]
and therefore
\begin{equation}\label{eq:2}
t_W<t_I.
\end{equation}
At the prevalence peak, $I'(t_I)=0$, so $A(t_I)=A(t_I-\tau)$. Since \(t_W<t_I\) and \(t^*\le t_W\), we have \(t_I>t^*\). If
\(t_I-\tau\geq t^*\), then both \(t_I-\tau\) and \(t_I\) lie on the
decreasing branch of the unimodal incidence curve, implying $A(t_I)<A(t_I-\tau)$,
a contradiction. Hence
\begin{equation}\label{eq:3}
t_I-\tau<t^*.
\end{equation}
\end{proof}

\subsection{Additional results for Section~\ref{ssec:accuracy}}
\label{app:accr}

\subsubsection{Proof of Proposition~\ref{pro2}}
\begin{proof}
Note that
\[
W_{\max} = 1 - \frac{1+\ln\big( R_0 (1-\epsilon) \big)}{R_0}.
\]
Therefore,
\[
E = \frac{2 - \frac{2}{R_0} \bigg( 1 + \ln \big( R_0 (1-\epsilon) \big) \bigg)}{I_{\max}} - 1
\]
and inequality $E<\eta$ is equivalent to 
\[
1 - \frac{1}{R_0} \bigg( 1 + \ln \big( R_0 (1-\epsilon) \big) \bigg) < \frac{\eta+1}{2} I_{\max}.
\]
Since $I_{\max} \leq 1$ and assuming that $\epsilon$ is small, a necessary condition is
\[
f(R_0) > \frac{1-\eta}{2},
\]
where $f(R_0) = \frac{1+\ln(R_0)}{R_0}$. Note that $f(R_0)$ is a continuous and strictly decreasing function on $(1, +\infty)$ with $f(1)=1$ and $\lim_{R_0 \to \infty}f(R_0) = 0$. Then, according to the Intermediate Value Theorem, the above inequality is equivalent to
\[
R_0 < R_0^{\#}.
\]
\end{proof}

\subsubsection{Proof of Proposition~\ref{th:R1}}
\begin{proof} Since $t_I > \tau$, we get
\be
\label{eq:t1}
I_{\max} = S(t_I - \tau) - S(t_I).
\ee
Since $t_I \notin \{\tau, 2\tau, .., j\tau,.. \}$ it follows that 
\[
\dot{I}(t_I) = 0 = \dot{S}(t_I - \tau) - \dot{S}(t_I) = -\beta S(t_I - \tau) I(t_I - \tau) + \beta S(t_I) I_{\max},
\]
which gives
\be
\label{eq:t2}
I_{\max} = \frac{S(t_I - \tau) I(t_I - \tau)}{S(t_I)}.
\ee
Since $t_I > 2 \tau$, we get
\be
\label{eq:t3}
I(t_I - \tau) =  S(t_I - 2\tau) - S(t_I - \tau) \leq 1 - S(t_I - \tau).
\ee
Plugging~\eqref{eq:t1} and~\eqref{eq:t3} into~\eqref{eq:t2} gives
\[
S^2(t_I - \tau) + \big( S(t_I)-1 \big) S(t_I - \tau) - S^2(t_I) \leq 0,
\]
which combining with the condition $S(t_I - \tau) \geq 0$ is equivalent to
\[
0 \leq S(t_I - \tau) \leq \frac{1-S(t_I) +\sqrt{1-2S(t_I)+5S^2(t_I)} }{2}.
\]
Plugging this inequality into~\eqref{eq:t1} gives
\be
\label{eq:t4}
I_{\max} \leq G\big( S(t_I) \big).
\ee
Under the proposition assumptions, $S(t_I) > \frac{1}{R_0} - \psi >0$. Since the function $G(\cdot)$ is monotonically decreasing on $(0,1)$,~\eqref{eq:t4} transforms into
\[
I_{\max} < G\bigg( \frac{1}{R_0} - \psi \bigg).
\]
Finally, with $\epsilon \to 0$, inequality $E < \eta$ gives 
\[
1 - \frac{1+\ln(R_0)}{R_0} < \frac{\eta+1}{2} I_{\max} < \frac{\eta+1}{2} G\bigg( \frac{1}{R_0} - \psi \bigg),
\]
which proves the proposition.
\end{proof}

\subsubsection{Proof of Proposition~\ref{th:psi}}
\begin{proof} Let $r = \frac{S(t_I-\tau)}{S(t_I)}$. Note that $r>1$ since $S$ is a decreasing function. First, we will represent $I_{\max}$ and $I(t_I-\tau)$ as functions of $r$. Since $t_I>2\tau$,
\be
\label{eq:a1}
I_{\max} = I(t_I) = S(t_I-\tau) - S(t_I) = S(t_I)(r-1).
\ee
Since $t_I\notin\{\tau,2\tau,3\tau,\ldots\}$,
\be
\label{eq:a2}
\dot{I}(t_I) = 0 \Leftrightarrow S(t_I-\tau)I(t_I-\tau) = S(t_I)I(t_I) \Leftrightarrow I(t_I - \tau) = \frac{I_{\max}}{r}.
\ee

Our main goal is to estimate the quantity $\int_{t_I-\tau}^{t_I} I(u)du$. Note that since $S$ is a decreasing function, for each $t \in [t_I-\tau, t_I]$:
\[
\dot{I}(t) = \beta S(t) I(t) - \beta S(t-\tau) I(t-\tau) \leq \beta S(t) I(t) \leq \beta r S(t_I) I(t).
\]
Therefore, applying the Grönwall's inequality, we get:
\[
I(t) \leq I(t_I-\tau) e^{\beta r S(t_I)\int_{t_I-\tau}^{t} du}, 
\]
which combining with~\eqref{eq:a2} gives
\[
I(t) \leq \frac{I_{\max}}{r} e^{\beta r S(t_I)(t-t_I+\tau)}.
\]
Since in addition $I(t) \leq I_{\max}$, we finally get
\[
I(t) \leq \min \bigg\{ I_{\max}, \frac{I_{\max}}{r} e^{\beta r S(t_I)(t-t_I+\tau)} \bigg\}.
\]
Given that $t' = \frac{\ln(r)}{\beta r S(t_I)} + t_I -\tau \in (t_I-\tau,t_I)$ (see Lemma~\ref{l:a1}),
\[
\int_{t_I-\tau}^{t_I} I(u)du \leq \int_{t_I-\tau}^{t'} \frac{I_{\max}}{r} e^{\beta r S(t_I)(u-t_I+\tau)}du + \int_{t'}^{t_I} I_{\max}du \leq \tau I_{\max} - \frac{I_{\max}}{\beta r S(t_I)}\Big(\frac{1}{r}-1+\ln(r)\Big).
\]
However since $\dot{S} = -\beta S(t)I(t)$, it follows that $\ln(r) = \beta \int_{t_I-\tau}^{t_I} I(u)du$. Plugging this into the previous inequality, we get
\[
\ln(r) \leq \beta \tau I_{\max} - \frac{I_{\max}}{r S(t_I)}\Big(\frac{1}{r}-1+\ln(r)\Big).
\]
Using~\eqref{eq:a1} and the fact that $R_0 = \beta \tau$, the above inequality can be rewritten as:
\be
\label{eq:a3}
S(t_I) \geq \frac{a(r)}{R_0}.
\ee

Finally, let's derive an upper bound for $r$. First, note that combining~\eqref{eq:a1}-\eqref{eq:a2}, we get:
\[
S(t_I - 2\tau) - S(t_I - \tau) = I(t_I - \tau) = S(t_I)(1-1/r).
\]
Therefore,
\[
S(t_I - 2\tau) = S(t_I)(r+1-1/r).
\]
Combining with~\eqref{eq:a3}, this gives
\[
a(r) \bigg(r +1 -\frac{1}{r} \bigg) \leq R_0.
\]
Following Lemma~\ref{l:a3} (see below) it implies $r \leq r^*(R_0)$, where $r^*(R_0)$ is the unique solution to~\eqref{eg:a01} on $(1,+\infty)$. According to Lemma~\ref{l:a2} below, $a(r)$ is decreasing on $(1,+\infty)$ and hence $a(r) \geq a(r^*(R_0))$. As a result: 
\[
S(t_I) \geq \frac{a(r^*(R_0))}{R_0},
\]
which establishes  the proposition once the Lemmas \ref{l:a1}--\ref{l:a3} are verified.
\end{proof}


\begin{lemma} 
\label{l:a1}
Assume that $t_I>2\tau$, $t_I\notin\{\tau,2\tau,3\tau,\ldots\}$. Let $r = \frac{S(t_I-\tau)}{S(t_I)}$. Then,
\[
t' = \frac{ln(r)}{\beta r S(t_I)} + t_I -\tau \in (t_I-\tau,t_I).
\]
\end{lemma}
\begin{proof} 
First, note that $t'$ is the unique point of intersection of the functions $f_1(t) = I_{\max}$ and $f_2(t) = \frac{I_{\max}}{r} e^{\beta r S(t_I)(t-t_I+\tau)}$. Second, note that 
\[
f_2(t_I-\tau) = \frac{I_{\max}}{r} < I_{\max}
\]
and
\[
f_2(t_I) = \frac{I_{\max}}{r} e^{\beta r S(t_I) \tau}.
\]
Since $\dot{S} = -\beta S(t)I(t)$ and employing~\eqref{eq:a1}, it follows that 
\[
\ln(r) = \beta \int_{t_I-\tau}^{t_I} I(u)du \leq \beta I_{\max} \tau = \beta (r-1) S(t_I) \tau < \beta r S(t_I) \tau,
\]
which implies
\[
r < e^{\beta r S(t_I) \tau},
\]
hence
\[
f_2(t_I) > I_{\max}.
\]
Since $f_2$ is continuous on $(t_I-\tau,t_I)$, the Intermediate Value Theorem guarantees that $t' \in (t_I-\tau,t_I)$, which proves the lemma.
\end{proof}

\begin{lemma}
\label{l:a2}
The function
\[
a(r)= \frac{\ln(r)+\frac{r-1}{r}\left(\ln(r)-1+\frac{1}{r}\right)}{r-1}
\]
is strictly decreasing for  $r>1$.
\end{lemma}

\begin{proof}
First, rewrite $a(r)$ as
\[
a(r) = \frac{(2r-1)\ln r}{r(r-1)} - \frac{r-1}{r^2}.
\]
Differentiating gives
\[
a'(r) = \frac{3r^3-7r^2+6r-2-r(2r^2-2r+1)\ln r}{r^3(r-1)^2}.
\]
Since the denominator is positive for $r>1$, it remains to show that the numerator is negative. Define
\[
B(r) = \ln r - \frac{3r^3-7r^2+6r-2}{r(2r^2-2r+1)}.
\]
Then the numerator of $a'(r)$ can be written as $-r(2r^2-2r+1)B(r)$. As a result, it is enough to show that $B(r)>0$ for all $r>1$. To do this, calculate
\[
B'(r) = \frac{ (r-1) (4r^4-12r^3+14r^2-7r+2)} {r^2(2r^2-2r+1)^2}.
\]
The denominator is positive for $r>1$. Note that
\[
4r^4-12r^3+14r^2-7r+2 = 4(r-1)^4+4(r-1)^3+2(r-1)^2+(r-1)+1>0.
\]
Therefore $B'(r)>0$ for all $r>1$. Since $B(1)=0$, it follows that
\[
B(r)>0 \qquad \text{for all } r>1,
\]
which proves the lemma.
\end{proof}

\begin{lemma}
\label{l:a3}
The function
\[
F(r)=a(r)\bigg(r+1-\frac{1}{r}\bigg)
\]
is strictly increasing on $(1,\infty)$. Moreover,
\[
\lim_{r\to 1^+}F(r)=1, \qquad \lim_{r\to\infty}F(r)=+\infty.
\]
\end{lemma}

\begin{proof}
First note that
\be
\label{eq:aa1}
a(r)= \frac{(2r-1)\ln r}{r(r-1)} - \frac{r-1}{r^2}.
\ee
Therefore
\[
F(r) = \frac{(r^2+r-1)(2r-1)\ln r}{r^2(r-1)} - \frac{(r^2+r-1)(r-1)}{r^3}.
\]
Differentiating gives
\[
F'(r) = \frac{2r^5-r^4-8r^3+15r^2-11r+3-r(3r^3-6r^2+6r-2)\ln r}{r^4(r-1)^2}.
\]
Since the denominator is positive for $r>1$, it remains to prove that the numerator is positive. To do this, set $x=r-1>0$. Then the numerator of $F'(r)$ becomes
\[
x+5x^2+8x^3+9x^4+2x^5 - (1+4x+6x^2+6x^3+3x^4)\ln(1+x).
\]
Note that it can be rewritten as $(1+4x+6x^2+6x^3+3x^4) Z(x)$, where
\[
Z(x)= \frac{x+5x^2+8x^3+9x^4+2x^5}{1+4x+6x^2+6x^3+3x^4} - \ln(1+x).
\]
As a result, $F'(r)>0$ for all $r>1$ if and only if $Z(x)>0$ for all $x>0$. Direct differentiation of $Z(x)$ yields
\[
Z'(x)= \frac{x (3+20x+66x^2+125x^3+141x^4+104x^5+54x^6+21x^7+6x^8)}{(1+x)(1+4x+6x^2+6x^3+3x^4)^2}.
\]
Therefore, $Z'(x)>0$ for all $x>0$. Since $\lim_{x \to 0^+} Z(x)=0$, it follows that $Z(x)>0$ for all $x>0$. As a result, $F$ is strictly increasing on $(1,\infty)$.

It remains to compute the limit of $F$ as $r \to \infty$. Using~\eqref{eq:aa1}, we get
\[
a(r)\sim \frac{2\ln r-1}{r}.
\]
In addition, $r+1-\frac{1}{r}\sim r$. Therefore,
\[
F(r) = a(r)\bigg(r+1-\frac{1}{r}\bigg)
\sim
2\ln r-1,
\]
and hence $\lim_{r\to\infty}F(r)=+\infty$, which completes the proof.
\end{proof}

\subsection{Additional results for Section~\ref{ssec:refined}}

\subsubsection{Derivation of the first-order approximation (FO)}
\label{app:FO_derivation}

Here we derive the first-order approximation
\[
\widehat I_{\max}^{\mathrm{FO}} \approx 2W_{\max}\left(1-\frac{\kappa\tau^2}{72}\right).
\]
Using $e^{-z^2/2}=1-\frac{z^2}{2}+O(z^4)$, we obtain
\[
J_I(\lambda) = \int_{-\lambda/2}^{\lambda/2}\left(1-\frac{z^2}{2}+O(z^4)\right)dz =
\lambda-\frac{\lambda^3}{24}+O(\lambda^5).
\]
Next, we denote $r=\lambda y$. Then
\[
J_W(\lambda) = \max_x \lambda \int_0^1 (1-y)e^{-(x-\lambda y)^2/2}\,dy .
\]
For $\lambda\ll 1$, the maximizer satisfies $x=O(\lambda)$, so
\[
e^{-(x-\lambda y)^2/2} = 1-\frac{(x-\lambda y)^2}{2} + O(\lambda^4).
\]
Therefore,
\[ 
J_W(\lambda) = \lambda\int_0^1(1-y)\,dy - \frac{\lambda}{2} \min_x \int_0^1(1-y)(x-\lambda y)^2\,dy + O(\lambda^5).
\]
The minimizer satisfies
\[
x\int_0^1(1-y)\,dy = \lambda\int_0^1y(1-y)\,dy,
\]
and hence $x=\frac{\lambda}{3}$. Substituting this value gives
\[
J_W(\lambda) = \frac{\lambda}{2} - \frac{\lambda^3}{72} + O(\lambda^5).
\]
Combining the two expansions gives
\[
C(\lambda) = 2 \frac{1-\lambda^2/24+O(\lambda^4)} {1-\lambda^2/36+O(\lambda^4)} =
2\left(1-\frac{\lambda^2}{72}\right) + O(\lambda^4).
\]
and
\[
\widehat I_{\max}^{\mathrm{FO}} \approx 2W_{\max} \left(1-\frac{\kappa\tau^2}{72}\right).
\]

\subsubsection{Derivation of the large-$\lambda$ approximations $L(0)$ and $L(1)$}
\label{app:L_derivation}

Here we derive the large-$\lambda$ approximations
\[
\widehat I_{\max}^{L(m)}
=
\Bigg[
\frac{\sqrt{2\pi}\left[2\Phi(\lambda/2)-1\right]}
{e^{-(x_W^{(m)})^2/2}
\left(\lambda-x_W^{(m)}-\lambda^{-1}\right)}
\Bigg]W_{\max},
\qquad m=0,1.
\]
For $J_I(\lambda)$, we have exactly
\[
J_I(\lambda)
=
\int_{-\lambda/2}^{\lambda/2}e^{-z^2/2}\,dz
=
\sqrt{2\pi}\left[2\Phi(\lambda/2)-1\right].
\]
It remains to approximate $J_W(\lambda)$. Let
\[
F(x)
=
\int_0^\lambda
\left(1-\frac{r}{\lambda}\right)
e^{-(x-r)^2/2}\,dr .
\]
Then $J_W(\lambda)=\max_x F(x)$. Setting $z=x-r$, we can rewrite
\[
F(x)
=
\int_{x-\lambda}^{x}
\left(1-\frac{x-z}{\lambda}\right)e^{-z^2/2}\,dz .
\]
Differentiating with respect to $x$ gives
\[
F'(x)
=
e^{-x^2/2}
-
\frac{1}{\lambda}
\int_{x-\lambda}^{x}e^{-z^2/2}\,dz .
\]
Therefore, the maximizer $x_W$ satisfies
\[
\lambda e^{-x_W^2/2}
=
\int_{x_W-\lambda}^{x_W}e^{-z^2/2}\,dz .
\]
For large $\lambda$, the lower limit $x_W-\lambda$ is far in the left tail, and therefore
\[
\int_{x_W-\lambda}^{x_W}e^{-z^2/2}\,dz
\approx
\sqrt{2\pi}\Phi(x_W).
\]
Thus $x_W$ approximately satisfies
\[
\lambda e^{-x_W^2/2}
\approx
\sqrt{2\pi}\Phi(x_W),
\]
or equivalently
\[
x_W
\approx
\left[
2\ln\left(
\frac{\lambda}{\sqrt{2\pi}\Phi(x_W)}
\right)
\right]^{1/2}.
\]

This motivates an iterative approximation for $x_W$. The zeroth-order approximation is obtained by setting $\Phi(x_W)\approx 1$, giving
\[
x_W^{(0)}
=
\left[
2\ln\left(
\frac{\lambda}{\sqrt{2\pi}}
\right)
\right]^{1/2}.
\]
This requires $\lambda>\sqrt{2\pi}$. Substituting $x_W^{(0)}$ once into the right-hand side gives the one-step approximation
\[
x_W^{(1)}
=
\left[
2\ln\left(
\frac{\lambda}{\sqrt{2\pi}\Phi(x_W^{(0)})}
\right)
\right]^{1/2}.
\]

Finally, we approximate $J_W(\lambda)=F(x_W)$. Using
\[
F(x)
=
\left(1-\frac{x}{\lambda}\right)
\int_{x-\lambda}^{x}e^{-z^2/2}\,dz
+
\frac{1}{\lambda}
\int_{x-\lambda}^{x}z e^{-z^2/2}\,dz ,
\]
and
\[
\int_{x-\lambda}^{x}z e^{-z^2/2}\,dz
=
e^{-(x-\lambda)^2/2}-e^{-x^2/2},
\]
we obtain, for large $\lambda$,
\[
F(x)
\approx
\left(1-\frac{x}{\lambda}\right)
\int_{x-\lambda}^{x}e^{-z^2/2}\,dz
-
\frac{1}{\lambda}e^{-x^2/2}.
\]
At the maximizer, we use
\[
\int_{x_W-\lambda}^{x_W}e^{-z^2/2}\,dz
\approx
\lambda e^{-x_W^2/2}.
\]
Therefore,
\[
J_W(\lambda)
=
F(x_W)
\approx
e^{-x_W^2/2}
\left(\lambda-x_W-\lambda^{-1}\right).
\]
Replacing $x_W$ by $x_W^{(m)}$, $m=0,1$, gives
\[
J_W(\lambda)
\approx
e^{-(x_W^{(m)})^2/2}
\left(\lambda-x_W^{(m)}-\lambda^{-1}\right).
\]
Hence
\[
C(\lambda)
=
\frac{J_I(\lambda)}{J_W(\lambda)}
\approx
\frac{\sqrt{2\pi}\left[2\Phi(\lambda/2)-1\right]}
{e^{-(x_W^{(m)})^2/2}
\left(\lambda-x_W^{(m)}-\lambda^{-1}\right)}.
\]
Thus,
\[
\widehat I_{\max}^{L(m)}
=
\Bigg[
\frac{\sqrt{2\pi}\left[2\Phi(\lambda/2)-1\right]}
{e^{-(x_W^{(m)})^2/2}
\left(\lambda-x_W^{(m)}-\lambda^{-1}\right)}
\Bigg]W_{\max},
\qquad m=0,1.
\]
The case $m=0$ gives the zeroth-order large-$\lambda$ approximation $L(0)$, while $m=1$ gives the one-step refinement $L(1)$.

\subsubsection{Proof of Lemma~\ref{th:lam}}
\label{app:plam}
\begin{proof}
Since $t^*>\tau$, we have
\[
I(t)=\int_{t-\tau}^{t}A(s)ds \quad\text{implying}\quad  I'(t)=A(t)-A(t-\tau).
\]
In addition,
\be
\label{eq:b1}
\frac{d}{dt}\ln A(t) = \frac{S'(t)}{S(t)} + \frac{I'(t)}{I(t)} = -\beta I(t)+\frac{I'(t)}{I(t)} = -\beta I(t) + \frac{A(t)-A(t-\tau)}{I(t)}.
\ee
At $t=t^*$, the incidence reaches its maximum, so
\[
\left.\frac{d}{dt}\ln A(t)\right|_{t=t^*}=0.
\]
Hence
\be
\label{eq:b3}
-\beta I(t^*) + \frac{A(t^*)-A(t^*-\tau)}{I(t^*)} = 0, \quad \text{and thus}\quad  A(t^*-\tau) = \beta I(t^*)\bigl(S(t^*)-I(t^*)\bigr).
\ee
Note that this implies $I(t^*)\leq S(t^*)$. Next note that
\[
\frac{d^2}{dt^2}\ln A(t) = -\beta I'(t) + \frac{I''(t)}{I(t)} - \left(\frac{I'(t)}{I(t)}\right)^2.
\]
At $t=t^*$, we have
\[
I'(t^*)=\beta I(t^*)^2.
\]
Also since $I''(t)=A'(t)-A'(t-\tau)$ and $A'(t^*)=0$, we get:
\[
I''(t^*)=-A'(t^*-\tau).
\]
Therefore,
\be
\label{eq:b2}
\kappa = -\left.\frac{d^2}{dt^2}\ln A(t)\right|_{t=t^*} = 2\beta^2 I(t^*)^2 + \frac{A'(t^*-\tau)}{I(t^*)}.
\ee

First, we bound the second term. Note that
\[
A'(t^*-\tau) = A(t^*-\tau) \left.\frac{d}{dt}\ln A(t)\right|_{t=t^*-\tau}.
\]
Employing~\eqref{eq:b1} and using the fact that $A(t)=\beta S(t)I(t)$, we obtain
\[
\frac{d}{dt}\ln A(t) = \beta\bigl(S(t)-I(t)\bigr) - \frac{A(t-\tau)}{I(t)} \leq \beta S(t) \leq \beta,
\]
which implies
\[
A'(t^*-\tau)\leq \beta A(t^*-\tau).
\]
Using~\eqref{eq:b3}, we conclude
\[
\frac{A'(t^*-\tau)}{I(t^*)}
\leq
\beta^2\bigl(S(t^*)-I(t^*)\bigr).
\]
Plugging this into~\eqref{eq:b2}, we get
\be
\label{eq:b4}
\kappa \leq \beta^2 \bigl(2I(t^*)^2+S(t^*)-I(t^*)\bigr).
\ee

Second, it remains to bound the expression in parentheses. We already showed that $I(t^*)\leq S(t^*)$. Also, $S(t^*)+I(t^*)\leq 1$. Finally, since $A(t)$ is strictly unimodal and reaches its maximum at $t^*$, it is increasing for $t^*-\tau\leq s\leq t^*$. Consequently, $A(s)\geq A(t^*-\tau)$ for every $s$ satisfying $t^*-\tau\leq s\leq t^*$. Therefore using~\eqref{eq:b3}, we conclude
\[
I(t^*) = \int_{t^*-\tau}^{t^*}A(s)ds \geq \tau A(t^*-\tau) \geq \tau\beta I(t^*)\bigl(S(t^*)-I(t^*)\bigr).
\]
Consequently since $I(t^*)>0$,
\[
S(t^*)-I(t^*)\leq \frac{1}{\tau \beta}.
\]

To get an upper bound for $\bigl(2I(t^*)^2+S(t^*)-I(t^*)\bigr)$, we can consider the following optimization problem:
\[
F= 2I^2+S-I \rightarrow \max_{I,S}
\]
\[
s.t.
\]
\[
I\leq S, \qquad S+I\leq 1, \qquad S-I\leq \frac{1}{\tau \beta}.
\]
Function $F$ is convex on its domain and for fixed $I$, it increases with $S$. As a result and since $R_0 \geq 1$, it can attain its maximum in one of the three points $(I,S)$: $p_0 = (0,1/(\tau \beta))$, $p_1 = (0.5 - 0.5/(\tau \beta), 0.5 + 0.5/(\tau \beta))$ or $p_2 = (0.5,0.5)$. Note that $F(p_0) = \frac{1}{\tau \beta}$, $F(p_1) = \frac{1+1/(\tau \beta)^2}{2}$, $F(p_2)=0.5$, $F(p_1) > F(p_0)$, and $F(p_1) > F(p_2)$. Consequently,
\[
\kappa \leq \beta^2 \cdot \frac{1+1/(\tau \beta)^2}{2} = \frac{(\tau \beta)^2+1}{2 \tau^2}.
\]
As a result,
\[
\lambda^2 = \tau^2\kappa \leq \frac{(\tau \beta)^2+1}{2} = \frac{R_0^2+1}{2}.
\]
\end{proof}

\subsubsection{Additional numerical results}
\begin{figure}[ht!]
    \centering
    \includegraphics[width=1\linewidth]{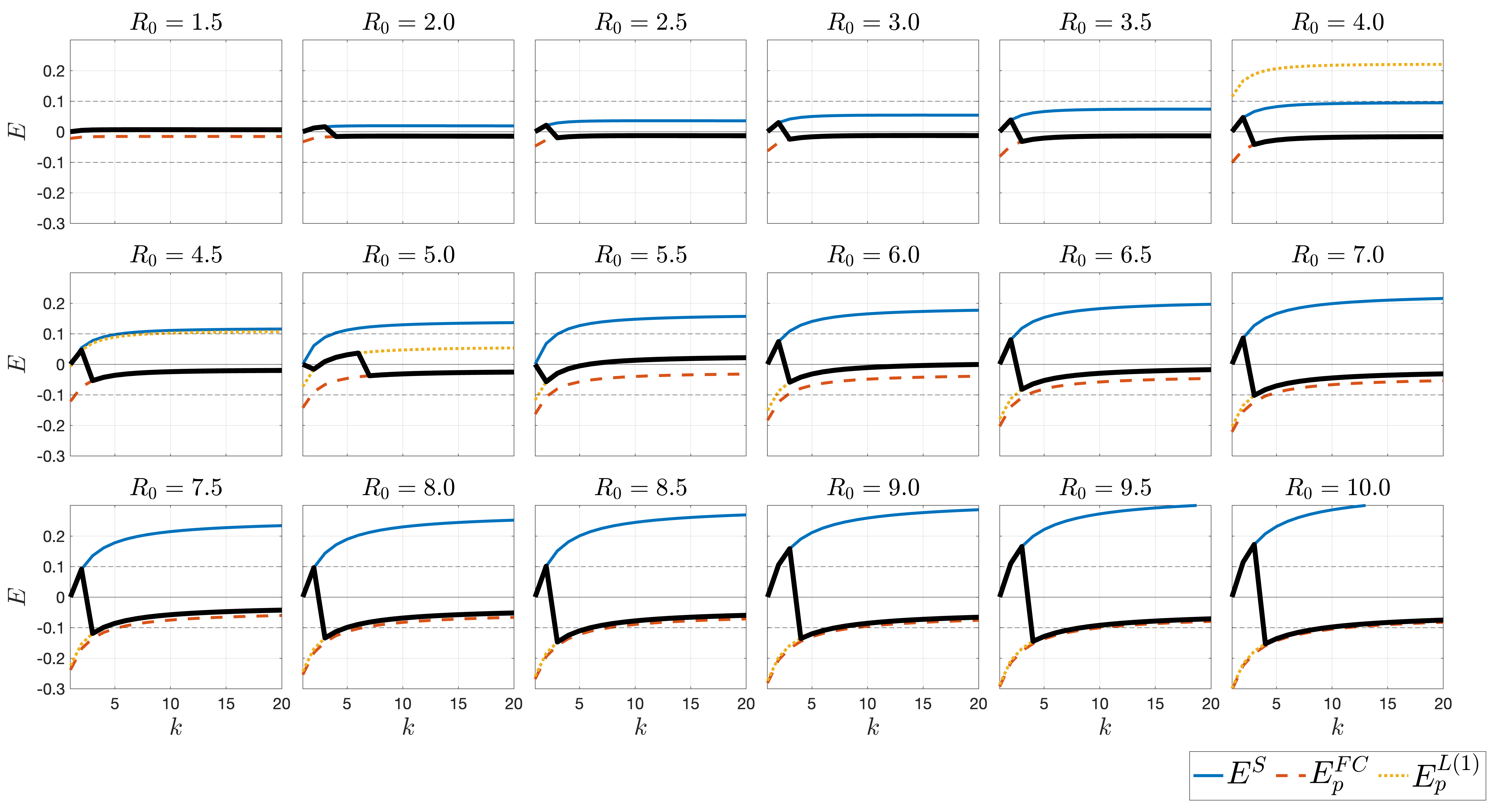}
    \caption{
    Relative errors of approximations of $I_{\max}^{(k)}$ in the finite $SI(k)R$ model as functions of the number of infectious stages $k$, shown for different values of $R_0$. The figure compares the simple approximation (S), the fully corrected plug-in approximation (FC,p) using $\lambda=\sqrt{(R_0^2+1)/2}$, and the corresponding large-$\lambda$ one-step approximation ($L(1),p$). The scaled weighted peak is $W_{\max}=V_{\max}^{(k)}/(k+1)$, and the error is $E=\widehat I_{\max}^{(k)}/I_{\max}^{(k)}-1$, with positive values indicating overestimation. Black lines indicate the approximation with the smallest absolute error for each $k$. Horizontal dashed lines mark $E=0$ and the $\pm0.1$ thresholds. Parameters are $\tau=5$ and $\varepsilon=0.01$. This should be compared with Figure~\ref{fig:5}, where $\tau=10$.
    } \label{fig:6}
\end{figure}

\end{appendices}
\clearpage


\end{document}